\documentclass[preprint,12pt]{elsarticle}

\usepackage{fixmath}
\usepackage{bm}
\usepackage{amsbsy}
\usepackage{color}
\usepackage{verbatim}
\usepackage{multirow}
\usepackage{amssymb}
\usepackage{amsthm}
\usepackage{array}
\usepackage{mathtools}
\usepackage{amsmath}
\usepackage{graphicx}
\usepackage{tikz}
\usetikzlibrary{positioning}
\usepackage{xcolor}
\usepackage{thm-restate}

\usepackage{graphicx}
\usepackage{subcaption}

\newtheorem{theorem}{Theorem}
\newtheorem{lemma}[theorem]{Lemma}

\newtheorem{definition}[theorem]{Definition}

\newcommand{\size}[1]{\ensuremath{|#1|}}
\newcommand{\ceil}[1]{\ensuremath{\lceil#1\rceil}}

\newcommand{\lrA}[1]{\ensuremath{\left(#1\right)}}

\let\epsilon=\varepsilon

\def\OPT{\mbox{OPT}}

\def\C{\mathcal{C}}

\def\P{\mathcal{P}}
\def\T{\mathcal{T}}

\def\I{\mathcal{I}}

\def\OPT{\mbox{OPT}}

\def\C{\mathcal{C}}

\def\P{\mathcal{P}}

\usepackage{algorithm}
\usepackage{algorithmicx}
\usepackage{algpseudocode}
\usepackage{url}
\usepackage{hyperref}
\usepackage{lineno}

\journal{}

\begin{document}

\begin{frontmatter}



\title{Enhanced Approximation Algorithms for the Capacitated Location Routing Problem} %% Article title


\author{Jingyang Zhao}
%\ead{jingyangzhao1020@gmail.com}
\author{Mingyu Xiao\corref{cor1}}
\ead{myxiao@gmail.com}
\cortext[cor1]{Corresponding author.}
\author{Shunwang Wang}
\affiliation
{
    organization={School of Computer Science and Engineering, University of Electronic Science and Technology of China},
    addressline={2006 Xiyuan Ave}, 
    city={Chengdu},
    postcode={610054}, 
    state={Sichuan},
    country={China}
}

\begin{abstract}
The Capacitated Location Routing Problem is an important planning and routing problem in logistics, which generalizes the capacitated vehicle routing problem and the uncapacitated facility location problem. In this problem, we are given a set of depots and a set of customers where each depot has an opening cost and each customer has a demand. The goal is to open some depots and route capacitated vehicles from the opened depots to satisfy all customers' demand, while minimizing the total cost. In this paper, we propose a $4.169$-approximation algorithm for this problem, improving the best-known $4.38$-approximation ratio. Moreover, if the demand of each customer is allowed to be delivered by multiple tours, we propose a more refined $4.091$-approximation algorithm. Experimental study on benchmark instances shows that the quality of our computed solutions is better than that of the previous algorithm and is also much closer to optimality than the provable approximation factor.
\end{abstract}


% \begin{highlights}
% \item A $4.169$-approximation algorithm for both splittable and unsplittable Capacitated Location Routing, improving the best-known approximation ratio of $4.38$ in (Transp. Sci. 2013).
% \item A $4.091$-approximation algorithm for splittable Capacitated Location Routing using path-splitting techniques, which partially addresses the open problem raised in (Transp. Sci. 2013) regarding whether a more tour-specific approach could lead to better approximation ratios.
% \item Experimental study on benchmark instances shows that the quality of our computed solutions is better than that of the previous algorithm and is also much closer to optimality than the provable approximation factor.
% \end{highlights}

%% Keywords
\end{frontmatter}

%\linenumbers

\section{Introduction}
In the fields of logistics, vehicle routing and facility location are two major problems that have been widely studied in both theory and application.
%Readers can refer to~\cite{toth2014vehicle} for a survey of CVRP. 
Given a set of depots, vehicle routing aims to route capacitated vehicles in the depots to deliver goods for customers to satisfy their demand using the minimum cost. Facility location concerns the opening cost of depots and the connection cost between customers and the opened depots. Location routing can be seen as a combination of these two problems, and is more natural in real life. It involves first opening a set of depots at some cost, and then routing the vehicles from the opened depots.
Location routing problems have been studied for decades since the idea of combining vehicle routing and facility location was introduced in~\cite{von1961relationship,maranzana1964location,webb1968cost}.
Recent surveys of location routing problems can be found in~\cite{DBLP:journals/eor/ProdhonP14,DBLP:journals/eor/DrexlS15}. 

In Capacitated Location Routing (CLR), we are given an undirected complete graph $G=(V\cup U, E, w, \phi,d,k)$, where $V$ is the set of customers, $U$ is the set of depots (or facilities), and %there are two cost functions with one demand function: $w:E\rightarrow\mathbb{R}_{\geq 0}$ on the edges in $E$, $\phi:U\rightarrow\mathbb{R}_{\geq 0}$ on the depots in $U$, and $d:V\rightarrow\mathbb{R}_{\geq 0}$ on the customers in $V$. 
there is one travel cost function $w:E\rightarrow\mathbb{R}_{\geq 0}$ on edges, an opening cost function $\phi:U\rightarrow\mathbb{R}_{\geq 0}$ on depots, and a demand function $d:V\rightarrow\mathbb{R}_{\geq 0}$ on customers. 
Each depot $u\in U$ has an opening cost $\phi(u)$ and contains an unbounded fleet of vehicles with the same capacity $k\in \mathbb{R}_{>0}$ (one can image that each vehicle in the depot can be used for many times), each customer $v\in V$ has a demand $d(v)>0$, and we need to determine a set of depots $O\subseteq U$ to open and a set of tours $\I$ such that (1) each tour starts and ends at the same (opened) depot, (2) each tour delivers at most $k$ of demand to customers on the tour, and (3) the union of tours in $\I$ satisfies all customers' demand. The total cost is defined as $\sum_{I\in\I}w(I)+\sum_{u\in O}\phi(u)$, where $w(I)$, the cost of tour $I$, is defined to be  the total cost of edges in $I$, i.e., $w(I)=\sum_{e\in I}w(e)$. We consider two variants of CLR: \emph{splittable} and \emph{unsplittable}, where the demand of each customer is allowed to be delivered by multiple tours (resp., only one tour) in splittable (resp., unsplittable) CLR.
This definition also captures the case where each tour incurs an extra depot-dependent fixed cost~\cite{tuzun1999two,barreto2007using}, i.e., each vehicle departing from depot $u\in U$ incurs an additional cost of $F_u\in\mathbb{R}_{\geq 0}$. This can be represented by adding $F_u/2$ to the cost of all edges incident to $u$, as each tour originating at $u$ uses only two of these edges.

CLR generalizes many famous NP-hard problems. If $\phi\equiv0$, we can open all depots in $U$ at no cost, and in this case CLR is the Multidepot Capacitated Vehicle Routing Problem (MCVRP). 
If $\phi\equiv0$ and $k=\infty$, CLR becomes the metric $m$-depot traveling salesman problem (TSP).
Hence, CLR also includes the (single depot) Capacitated Vehicle Routing Problem (CVRP) and metric TSP as special cases.
Moreover, CLR with $k$ being the greatest common divisor of the demands also generalizes Uncapacitated Facility Location (UFL), where we choose to open some depots and assign each customer $v$ to its nearest opened depot $u$ with a cost of $d(v) w(v,u)$.
%(in this reduction $d(v)$ may be required to an integer).
%Moreover, if $d(v)\equiv k$ for every $v\in V$, CLR is known to be Uncapacitated Facility Location (UFL), where we choose to open some depots and directly assign each customer to its nearest depot.
These problems have been extensively studied both in terms of approximation algorithms and experimental algorithms~\cite{ravi2006approximation,toth2014vehicle,montoya2015literature,an2017lp,DBLP:conf/aaai/ZhangHLQG15,xin2021multi,DBLP:conf/icml/0002WSCZ23}.

\subsection{Related Work}
We focus on approximation algorithms.
Next, we give a brief review of literature on TSP, Vehicle Routing, UFL, and CLR.

\textbf{TSP.} For metric TSP, the Christofides-Serdyukov algorithm~\cite{christofides1976worst,serdyukov1978some} is a well-known $1.5$-approximation algorithm, and the ratio has been recently improved to $1.5-10^{-36}$ by~\cite{KarlinKG21,DBLP:conf/ipco/KarlinKG23}.
For metric $m$-depot TSP, Rathinam \emph{et al.}~\cite{rathinam2007resource} proposed a simple 2-approximation algorithm, and Xu \emph{et al.}~\cite{xu2011analysis} improved the ratio to $2-1/m$. There are also some works for the case that $m$ is fixed~\cite{xu20153,DBLP:journals/siamcomp/TraubVZ22,deppert20233}.
%Then, based on an edge exchange algorithm, Xu and Rodrigues~\cite{xu20153} obtained an improved $3/2$-approximation algorithm for any fixed $m$. Traub~\emph{et al.}~\cite{DBLP:journals/siamcomp/TraubVZ22} further improved the ratio to $\alpha+\varepsilon$ for any fixed $m$. Recently, Deppert~\emph{et al.}~\cite{deppert20233} obtained a randomized $(3/2+\varepsilon)$-approximation algorithm with a running time of $(1/\varepsilon)^{O(d\log d)}\cdot n^{O(1)}$, and hence their algorithm even works with a \emph{variable} number of depots.

\textbf{Vehicle Routing.} 
%CVRP was first raised by Dantzig and Ramser~\cite{dantzig1959truck}. 
%CVRP has numerous applications, and for the case of single depot with uncapacitated vehicles, CVRP is known as metric TSP. 
%For metric TSP, the Christofides-Serdyukov algorithm~\cite{christofides1976worst,serdyukov1978some} is a well-known $1.5$-approximation algorithm, and the ratio has been recently improved to $1.5-10^{-36}$ by Karlin \emph{et al.}~\cite{KarlinKG21,DBLP:conf/ipco/KarlinKG23}.
For the case of single depot, Haimovich and Kan~\cite{HaimovichK85} proposed a 2.5-approximation algorithm for splittable CVRP, and Altinkemer and Gavish~\cite{altinkemer1987heuristics} proposed a 3.5-approximation algorithm for unsplittable CVRP.
For the case of multidepot, Li and Simchi-Levi~\cite{tight} proposed a 4-approximation algorithm for splittable MCVRP, and Harks \emph{et al.}~\cite{HarksKM13} proposed a 4-approximation algorithm for unsplittable MCVRP.
These results got improved only very recently. 
For CVRP, Blauth \emph{et al.}~\cite{blauth2022improving} improved the ratio to $2.5-\frac{1}{3000}$ for the splittable case, and Friggstad \emph{et al.}~\cite{uncvrp} improved the ratio to about $3.164$ for the unsplittable case.
For MCVRP, Zhao and Xiao~\cite{DBLP:conf/cocoon/ZhaoX23} obtained a ratio of $4-\frac{1}{1500}$ for the splittable case and a ratio of $4-\frac{1}{50000}$ for the unsplittable case. Recently, Lai \emph{et al.}~\cite{DBLP:journals/informs/LaiX0D23} proposed a $(6-4/m)$-approximation algorithm for MCVRP with fixed $m$, considering constraints where the number of vehicles at each depot is limited and each vehicle can be used for only one tour.

\textbf{UFL.} 
UFL has a rather rich history on approximation algorithms (see the book~\cite{williamson2011design}), where many new techniques were developed. We mention the following results: Jain \emph{et al.}~\cite{DBLP:journals/jacm/JainMMSV03} proposed a practical 1.861-approximation algorithm based on the greedy method with a running time of $O(nm\log nm)$; Byrka and Aardal~\cite{DBLP:journals/siamcomp/ByrkaA10} proposed a 1.5-approximation algorithm and a bifactor approximation algorithm by modifying the LP rounding method in \cite{DBLP:journals/siamcomp/ChudakS03}; The current best result is a 1.488-approximation algorithm~\cite{DBLP:journals/iandc/Li13}.

\textbf{CLR.}
Since CLR is more challenging, there are only a few results on approximation algorithms.
Harks \emph{et al.}~\cite{HarksKM13} proposed a 4.38-approximation algorithm for both unsplittable and splittable CLR, and showed that unsplittable CLR cannot be approximated better than %by a factor of 
1.5 unless $P=NP$. They also extended their algorithm to derive approximation algorithms for three settings of CLR: prize-collecting, grouping, and cross-docking. Recently, Heine \emph{et al.}~\cite{ejor/HeineDM23} proposed a bifactor approximation algorithm for a variant of CLR, where each depot is also capacitated.

\subsection{Our Results}
In this paper, we propose two improved approximation algorithms for CLR. The first, denoted as Tree-Alg, is a 4.169-approximation algorithm for both unsplittable and splittable CLR, which improves the previous 4.38-approximation algorithm~\cite{HarksKM13}. Note that the previous approximation ratio has been kept for a decade. The second, denoted as Path-Alg, achieves a better ratio of 4.091 for splittable CLR. 
The outline of our improvements are as follows.

First, the main idea of the previous algorithm~\cite{HarksKM13} is as follows: given an instance $G$ of CLR, it constructs an instance $\hat{G}$ for UFL, opens some depots by calling the bifactor approximation algorithm for UFL~\cite{DBLP:journals/siamcomp/ByrkaA10} on $\hat{G}$, finds a set of desired trees, and finally satisfies the customers by splitting the trees.
When constructing the instance $\hat{G}$, the new opening cost of each depot $u\in U$ was set equal to $\phi(u)$. However, we find that this choice is not necessarily optimal. To address this, we set the new opening cost of each $u\in U$ to $\alpha\cdot\phi(u)$, where $\alpha$ is a new parameter that can be optimized in the final analysis. 
Based on the above observation and a better upper bound of the set of trees, we derive two stronger lower bounds for CLR, which are crucial to our results.
Moreover, by carefully tuning the parameter in the bifactor approximation algorithm for UFL, we obtain our Tree-Alg, which can be viewed as a refined version of the previous algorithm and improves the previous approximation ratio from $4.38$ to $4.169$.

Second, as asked in~\cite{HarksKM13}, an open problem is whether a more tour-specific approach could lead to better approximation ratios. We answer this question partially by showing that for splittable CLR our Path-Alg, which focuses on splitting paths, achieves a better approximation ratio of $4.091$. 
Our approach is inspired by the cycle-splitting method used for MCVRP in~\cite{tight}.
%motivated by the cycle-splitting method used in vehicle routing problems~\cite{tight}, we develop our Path-Alg, focusing on splitting paths.
Note that the idea of their cycle-splitting method is as follows: given an instance $G$ of CLR, it constructs a new graph $H$ by contracting all depots in $U$ as a super-depot, obtains a Hamiltonian cycle by calling a $\delta$-approximation algorithm for metric TSP on $H$, and then obtain a set of tours in $H$ by using the well-known cycle partition algorithm for CVRP~\cite{HaimovichK85,altinkemer1987heuristics}, which corresponds to a set of cycles in $G$. It then modifies the set of cycles into a set of feasible tours by introducing some additional costs (e.g., the cost of the Hamiltonian cycle used).
However, for CLR, we cannot directly use this idea. One reason is that, to measure the opening cost of depots, the edges incident to depots in the graph may have extra costs, and hence we need to carefully control the number of edges incident to depots. To do this, after obtaining a Hamiltonian cycle, we directly consider its corresponding edges in $G$ which in fact forms a set of cycles and paths, and then we obtain a set of paths in $G$ so that each path contains only one depot as one of its terminals by carefully shortcutting and deleting the edges incident to the depots. Finally, we satisfy customers by splitting the paths. 
Note that paths are simpler than trees in structure, and hence Path-Alg ensures that almost all tours deliver exactly $k$ of demand, whereas for Tree-Alg almost all tours deliver only $k/2$ of demand in the worst case. 
Although Path-Alg seems to perform better in the sense that it significantly reduces the connection cost by the above property, the analysis of Path-Alg uses more techniques. 
The main reason is that Path-Alg uses an approximation algorithm of TSP to compute a set of paths in $G$ which are more expensive, and then we cannot even obtain a better approximation ratio by a straightforward analysis. 

Last, we do experiments to evaluate our algorithms. In practice, our algorithms are easy to implement and run very fast.
Experimental study on benchmark instances shows that the quality of our computed solutions is better than that of the previous algorithm and is also much closer to optimality than the provable approximation factor.

\subsection{Paper Organization}
The remaining parts of the paper are organized as follows. In Section~\ref{sec2}, we introduce some notation and the formal definitions of CLR and UFL. 
In Section~\ref{sec3}, we propose two stronger lower bounds for CLR. 
In Section~\ref{sec4}, we give our Tree-Alg for both unsplittable and splittable CFL, and in Section~\ref{sec5}, we give our Path-Alg for splittable CFL. 
In Section~\ref{sec6}, we present the experimental study of our algorithms. 
At last, we make the concluding remarks in Section~\ref{sec7}.

\section{Preliminaries}\label{sec2}
In CLR, we use $G=(V\cup U, E, w, \phi, d, k)$ to denote the input complete graph. %, where the vertices in $V$ represent customers and the vertices in $U$ represent depots. There are two non-negative cost functions: $w: E\to \mathbb{R}_{\geq0}$ on the edges in $E$ and $\phi: U\to \mathbb{R}_{\geq0}$ on the depots in $U$. 
%We often write $w(u,v)$ to mean the cost of edge $uv$, instead of $w(uv)$.
The cost function $w$ is a metric function, i.e., it is symmetric and satisfies the triangle inequality.
%Each depot $u\in U$ has an opening cost of $\phi(u)$, and it contains an unbounded fleet of vehicles with a capacity of $k>0$.
%There is also a demand function $d$: $V\to\mathbb{R}_{\geq0}$, where $d(v)$ is the demand required by customer $v\in V$.
In UFL, the input graph is the same as CLR, except for the absence of the parameter $k$, and we use $\hat{G}=(V\cup U,$ $E, w, \phi, d)$ to denote it.

For any function $f: X\to \mathbb{R}_{\geq0}$, we always define $f(Y) = \sum_{x\in Y} f(x)$ for any $Y\subseteq X$.
For any subgraph $S$ of $G$, we use $V(S)$, $U(S)$, and $E(S)$ to denote the customer set, the depot set, and the edge set of $S$, respectively. Furthermore, we define $w(S)=\sum_{e\in E(S)}w(e)$, $\phi(S)=\sum_{u\in U(S)}\phi(u)$, and $d(S)=\sum_{v\in V(S)}d(v)$.

A \emph{walk} in a graph, denoted by $v_1v_2\dots v_l$, is a sequence of vertices $v_1,v_2,\dots,v_l$, where a vertex may appear more than once and each consecutive pair of vertices is connected by an edge.
%A \emph{walk} in a graph, denoted by $v_1v_2\dots v_l$, is a sequence of vertices $v_1, v_2, \dots, v_l$, where each vertex is connected to its adjacent vertices by edges in the graph, and vertices may appear more than once in the sequence.
%A \emph{walk} in a graph, denoted by $v_1v_2v_3\dots v_l$, is a sequence of edges $v_1v_2$, $v_2v_3$, and so on, where each edge may appear more than once. %It can be denoted by a sequence of vertices. For example, $v_1v_2v_3\dots v_l$ means a walk with edges $v_1v_2$, $v_2v_3$, and so on.
A \emph{path} in a graph is a walk where no vertex appears more than once. The first and the last vertices of a path are referred to as its \emph{terminals}.
A \emph{closed} walk is a walk where the first and the last vertices are the same, and
a \emph{cycle} is a walk where only the first and the last vertices are the same.
Given a closed walk, we can skip repeated vertices along the walk to get a cycle, and such an operation is called \emph{shortcutting}. 
During shortcutting, if a specific vertex $v$ is always skipped, we refer to it as \emph{shortcutting $v$}.

%Given a closed walk, we can skip repeated vertices along the walk to get a cycle, and such an operation is called \emph{shortcutting}. Moreover, if a specific vertex $v$ is skipped in a cycle obtained through shortcutting, we refer to it as \emph{shortcutting $v$}.
%A cycle containing $l$ edges is called an \emph{$l$-cycle} and the \emph{length} of it is $l$.
%An \emph{Eulerian} graph is a connected graph where every vertex has an even degree (i.e., an even number of edges incident to the vertex). It can be efficiently transformed into a cycle containing every vertex by shortcutting in polynomial time.

%For any walk $W$ and vertex $v$ of $W$, if we skip $v$ and keep other vertices of $W$ in the order as they appear in $W$ to obtain a new walk $W'$. Due to the triangle inequality, we have $w(P′) \leq w(P)$.

A \emph{constrained spanning forest} in $G$ is a forest that spans (i.e., covers) all vertices in $V$ and each tree in it contains only one depot.
A \emph{constrained spanning path-packing} in $G$ is a set of vertex-disjoint paths that spans all vertices in $V$, each depot is contained in at most one path in it, and each path in it contains only one depot and the depot is one of its terminals.
A \emph{constrained spanning cycle-packing} in $G$ is a set of vertex-disjoint cycles or paths that spans all vertices in $V$, each cycle in it contains only one depot, and each path in it contains only two depots and the depots are its terminals. An illustration of constrained spanning forest, path-packing and cycle-packing can be found in Figure~\ref{fig:whole_figure}. 

\begin{figure}[ht]
    \centering
    \begin{subfigure}[b]{0.48\textwidth}
        \centering
        \includegraphics[width=0.8\textwidth]{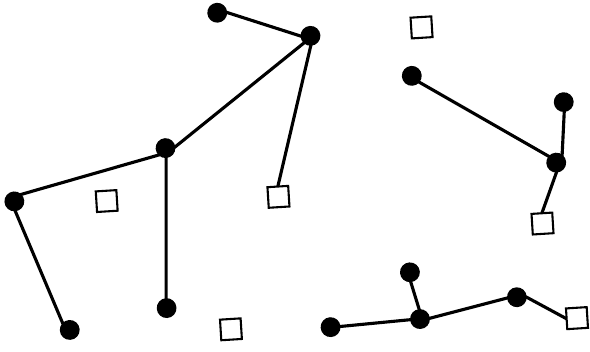}
        \caption{A constrained spanning forest.}
        \label{fig:fig1}
    \end{subfigure}
    
    \vspace{2mm}
    
    \begin{subfigure}[b]{0.48\textwidth}
        \centering
        \includegraphics[width=0.8\textwidth]{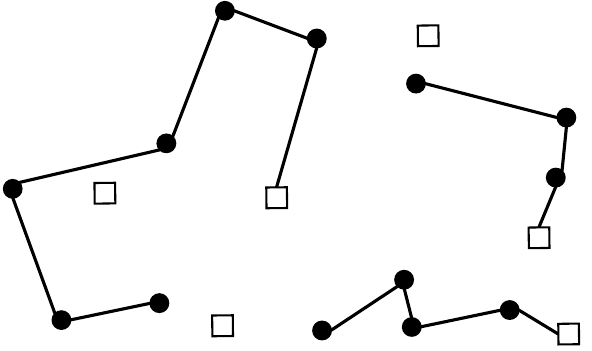}
        \caption{A constrained spanning path-packing.}
        \label{fig:fig2}
    \end{subfigure}
    \hfill
    \begin{subfigure}[b]{0.48\textwidth}
        \centering
        \includegraphics[width=0.8\textwidth]{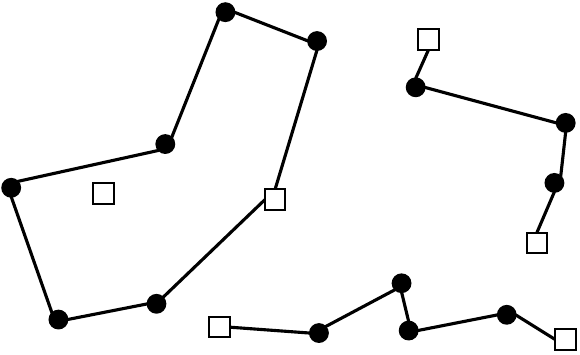}
        \caption{A constrained spanning cycle-packing.}
        \label{fig:fig3}
    \end{subfigure}
    \caption{An illustration of a constrained spanning forest, constrained spanning path-packing, and constrained spanning cycle-packing, where each cycle node represents a customer and each square node represents a depot.}
    \label{fig:whole_figure}
\end{figure}

Given a constrained spanning path-packing $\P$, for any path $P=u...b...c\in\P$, it contains a depot, say $u\in U$, as one of its terminals, and we refer to it as \emph{a path rooted at $u$}, and then $b...c$ is called \emph{the sub-path rooted at $b$}.
A \emph{tour} is a walk that starts and ends at the same depot and does not pass through any other depot. 
By the triangle inequality, we may restrict our attention to simple and minimal tours with each containing only one depot, i.e., each tour is a cycle containing only one depot.

\subsection{Formal Problem Definition}
%CLR can be described as follows.
\begin{definition}[\textbf{CLR}]
Given an undirected complete graph $G=(V\cup U, E, w, \phi, d,k)$, we need to find a set of depots $O\subseteq U$ and a set of tours $\I$ with a demand assignment $x: V\times\I\to \mathbb{R}_{\geq0}$ such that
\begin{enumerate}
\item[(1)] $U(I)\cap O\neq \emptyset$ for any $I\in\I$,
\item[(2)] $\sum_{v\in V(I)}x_{vI}\leq k$ for any $I\in \I$,
\item[(3)] $\sum_{v\in V\setminus V(I)}x_{vI}=0$ for any $I\in \I$,
\item[(4)] $\sum_{I\in\I}x_{vI}=d(v)$ for any $v\in V$,
\end{enumerate}
and $\sum_{I\in\I}w(I)+\sum_{u\in O}\phi(u)$ is minimized.
\end{definition}
Note that (1) ensures that each tour contains one depot, (2) and (3) ensures that each tour delivers at most $k$ of demand to those customers on the tour, and (4) ensures that each customer is fully 
 satisfied by the tours in $\I$.

In the above definition, if each customer is allowed to be satisfied by using multiple tours, we call it as \emph{splittable} CLR.
For \emph{unsplittable} CLR, it requires that each customer must be satisfied by using only one tour. Clearly, unsplittable CLR admits a feasible solution only if it holds $d(v)\leq k$ for any $v\in V$. In any solution $(O,\I)$, we will refer to $\sum_{I\in\I}w(I)$ as the \textit{routing cost} and $\sum_{u\in O}\phi(u)$ as the \textit{opening cost}.

%UFL is another problem that will be frequently mentioned. 

\begin{definition}[\textbf{UFL}]\label{UFL}
Given an undirected complete graph $\hat{G}=(V\cup U, E, w, \phi, d)$, we need to find a set of depots $O\subseteq U$ such that the cost $\sum_{v\in V}d(v)\cdot\min_{o\in O}w(o,v)+\sum_{u\in O}\phi(u)$ is minimized, i.e., for each $v\in V$ we directly assign $d(v)$ of demand to $v$ from its nearest depot $o$ in $O$ with a connection cost of $d(v)w(o,v)$.
\end{definition}

In Definition~\ref{UFL}, $\sum_{v\in V}d(v)\cdot\min_{o\in O}w(o,v)$ is called the \textit{connection cost} and $\sum_{u\in O}\phi(u)$ is called the \textit{opening cost}.

\section{Lower Bounds}\label{sec3}
In this section, we prove two lower bounds on the optimal solution 
that holds for both unsplittable and splittable CLR. These bounds are crucial for us to prove the approximation ratio. 

Given an instance $G=(V\cup U, E, w, \phi, d, k)$ of CLR, we construct an UFL instance $\hat{G}=(V\cup U, E, \widetilde{w}, \widetilde{\phi}, d)$ as follows. The sets of depots and customers with their demand remain the same as in CLR, but we set the costs of edges to $\widetilde{w}\coloneqq(2/k)w$ and the costs of depots to $\widetilde{\phi}\coloneqq \alpha\cdot\phi$, where $\alpha$ is a fixed constant that will be defined later. 
Here, $\alpha$ is an important parameter newly proposed by us. 
Previously, Harks \emph{et al.}~\cite{HarksKM13} focused solely on $\widetilde{\phi}\coloneqq \phi$, i.e., $\alpha=1$. However, in our setting, $\alpha$ may be greater than 1.

Let $\OPT$ (resp., $\OPT'$) denote the cost of an optimal solution for CLR (resp., UFL). We let $\psi^*$ and $\phi^*$ denote the routing cost of vehicles and the opening cost of facilities with respect to the optimal solution of CLR, respectively. Note that $\OPT=\psi^*+\phi^*$. We have the following bound.

\begin{lemma}\label{lb1}
It holds that $\OPT'+(1-\alpha)\cdot\phi^*\leq \OPT$.
\end{lemma}
\begin{proof}
Consider an optimal solution $(O,\I)$ of CLR with demand assignments $x_{vI}$. We have $\psi^*=w(\I)$ and $\phi^*=\phi(O)$.
As mentioned before, we construct an UFL instance using $\widetilde{w}\coloneqq(2/k)w$ and $\widetilde{\phi}\coloneqq \alpha\cdot\phi$. Next, we construct a feasible solution of UFL as follows. 

In the constructed instance of UFL, we open all depots in $O$. 
Clearly, the connection cost $\sum_{v\in V}d(v)\cdot\min_{o\in O}w(o,v)$ is minimized if for each $v\in V$  we directly assign $d(v)$ of demand from its nearest depot in $O$ to $v$. 
To obtain an upper bound of the connection cost, we consider the following assignments: for each tour $I\in\I$ and $v\in V(I)$ we assign $x_{vI}$ of demand from the depot $u_I$ in $I$ to $v$. By definition of CLR, the above assignments can satisfy all customers. Therefore, the cost of the constructed solution for UFL satisfies
\begin{align*}
\sum_{v\in V}d(v)\cdot\min_{o\in O}\widetilde{w}(o,v)+\sum_{u\in O}\widetilde{\phi}(u)&\leq\sum_{I\in\I}\sum_{v\in V(I)}x_{vI}\cdot \widetilde{w}(v,u_I)+\sum_{u\in O}\widetilde{\phi}(u)\\
&=\sum_{I\in\I}\sum_{v\in V(I)}x_{vI}\cdot(2/k)w(v,u_I)+\alpha\cdot\phi(O)\\
&\leq\sum_{I\in\I}\sum_{v\in V(I)}x_{vI}\cdot(1/k)w(I)+\alpha\cdot\phi(O)\\
&\leq\sum_{I\in\I}w(I)+\alpha\cdot\phi(O)\\
&=w(\I)+\alpha\cdot\phi(O)\\
&=\psi^*+\alpha\cdot\phi^*,
\end{align*}
where the second inequality follows from $2w(v,u_I)\leq w(I)$ by the triangle inequality and the third inequality follows from $\sum_{v\in V(I)}x_{vI}\leq k$ by definition of CLR.

Since $\OPT'$ is the cost of an optimal solution of UFL, we have $\OPT'\leq\psi^*+\alpha\cdot\phi^*$, and thus $\OPT'+(1-\alpha)\cdot\phi^*\leq \psi^*+\phi^*=\OPT$.
\end{proof}

Our bound in Lemma~\ref{lb1} is more general since the previous paper~\cite{HarksKM13} only obtained the bound under $\alpha=1$. Next, we consider the second lower bound.

It was shown in Lemma 2 in~\cite{HarksKM13} that one can use polynomial time to find a constrained spanning forest $\T$ in $G$ such that $2w(\T)+\phi(\T)\leq 2\cdot\OPT$. We propose a better result.

\begin{lemma}\label{lb2}
There is a polynomial-time algorithm to find a constrained spanning forest $\T$ such that $2w(\T)+\phi(\T)+\phi^*\leq 2\cdot\OPT$.
\end{lemma}
\begin{proof}
Given $G=(V\cup U, E, w, \phi)$, we obtain a new graph $G'=(V\cup U, E, w', \phi)$ such that $w'(v,v')=w(v,v')$ for any $v,v'\in V$ and $w'(u,v)=w(u,v)+\frac{1}{2}\phi(u)$ for any $v\in V$ and $u\in U$. As shown in~\cite{HarksKM13}, a minimum constrained spanning forest in $G'$, denoted by $\T$, can be found in polynomial time: first we can create a new vertex $r$ and set $w'(r,u)=0$ for each $u\in U$; then we find a minimum spanning tree in $G'[V\cup U\cup\{r\}]$, which can be done, for example, using the Prim's algorithm~\cite{cormen2022introduction}. Next, we bound $2w(\T)+\phi(\T)$.

Consider an optimal solution $(O,\I)$ of CLR. For each depot $o\in O$, there is a set of tours $\I_o\subseteq\I$ with each in it containing $o$. By shortcutting all tours in $\I_o$ and deleting an edge incident to $o$, we obtain a path $P_o=ov_1\dots v_i$ such that $w'(P_o)=w(P_o)+\frac{1}{2}\phi(o)\leq w(\I_o)+\frac{1}{2}\phi(o)$. Note that the paths in $\{P_o\}_{o\in O}$ may not be vertex-disjoint. Then, we can obtain a constrained spanning forest in $G'$ with a cost of at most $\sum_{o\in O}w'(P_o)$ by further shortcutting. Since $\T$ is a minimum constrained spanning forest in $G'$, we can get 
\begin{equation}\label{eq1}
\begin{split}
w'(\T)&\leq\sum_{o\in O}w'(P_o)\\
&=\sum_{o\in O}\lrA{w(\I_o)+\frac{1}{2}\phi(o)}\\
&\leq w(\I)+\frac{1}{2}\phi(O)\\
&=\psi^*+\frac{1}{2}\phi^*.
\end{split}
\end{equation}

For each $T\in\T$, let $i_T$ be the number of edges in $T$ that are incident to the depot in $U(T)$. Note that $i_T\geq 1$. We can get $w(T)=w'(T)-i_T\cdot \frac{1}{2}\phi(T)\leq w'(T)-\frac{1}{2}\phi(T)$. Hence, 
$
2w(\T)+\phi(\T)=\sum_{T\in\T}(2w(T)+\phi(T))\leq \sum_{T\in\T}2w'(T)=2w'(\T)\leq 2\psi^*+\phi^*,
$
where the second inequality follows from (\ref{eq1}).
%$2w(\T)+\phi(\T)=\sum_{T\in\T}(2w(T)+\phi(T))\leq \sum_{T\in\T}2w'(T)=2w'(\T)\leq 2\psi^*+\phi^*$.
Since $\psi^*+\phi^*=\OPT$, we have $2w(\T)+\phi(\T)+\phi^*\leq 2\cdot\OPT$.
\end{proof}

\section{General CLR}\label{sec4}
In this section, we introduce our Tree-Alg, which is a $4.169$-approximation algorithm for unsplittable and splittable CLR.

\subsection{The Algorithm}
Our algorithm uses the framework of the $4.38$-approximation algorithm in~\cite{HarksKM13}. It uses a bifactor approximation algorithm to compute a solution (i.e., a set of depots) for the constructed instance of UFL in Lemma~\ref{lb1}. Then, it computes the constrained spanning forest $\T$ in Lemma~\ref{lb2}. Let $O_1$ be the set of opened depots in the solution of UFL and $O_2$ be the set of depots contained in $\T$. The algorithm will open all depots in $O\coloneqq O_1\cup O_2$. Based on splitting the trees in $\T$ (the tree-splitting procedure in~\cite{HarksKM13}), the algorithm will return a feasible solution in polynomial time.
The framework of the algorithm is shown in Algorithm~\ref{algo:unsplittable UFL}. %the tree-splitting procedure is shown in Algorithm~\ref{algo:tree-splitting}.
Note that the 4.38-approximation algorithm in~\cite{HarksKM13} corresponds to our algorithm under $\alpha=1$.

\begin{algorithm}[H]
\caption{An improved approximation algorithm for unsplittable and splittable CLR (Tree-Alg)}
\label{algo:unsplittable UFL}
\small
\vspace*{2mm}
\textbf{Input:} An instance of CLR. \\
\textbf{Output:} A feasible solution to CLR.

\begin{algorithmic}[1]
\State Create an UFL instance with edge costs $\widetilde{w}=(2/k)w$ and depot costs $\widetilde{\phi}=\alpha\cdot\phi$ as in Lemma~\ref{lb1}.

\State Apply Byrka and Aardal's bifactor approximation algorithm~\cite{DBLP:journals/siamcomp/ByrkaA10} with a parameter of $\gamma>0$ on the UFL instance, and let $O_1$ be the set of depots opened in the resulting UFL solution.\label{bifactor}

\State Compute a constrained spanning forest $\T$ in $G$ as in Lemma~\ref{lb2}, and let $O_2=U(\T)$ be the set of depots contained in some $T\in\T$.

\State Open all depots in $O\coloneqq O_1\cup O_2$.

\State Obtain a set of feasible tours $\I$ by calling the tree-splitting procedure in Algorithm~\ref{algo:tree-splitting}.

\State Return $(O, \I)$.
\end{algorithmic}
\end{algorithm}

Given a set of opened depots $O$ and a constrained spanning forest $\T$, the tree-splitting procedure is to obtain a set of feasible tours using only depots in $O$, which is also equivalent to solving an instance of MCVRP. 
Note that it also works for the case of splittable demand. Hence, we may assume w.l.o.g. that a customer can have a demand of more than $k$.
The main idea is to repeatedly find a sub-tree $S$ of trees in $\T$ such that $k\geq d(S)>k/2$, and then construct a tour for some of the customers (as the sub-trees are only edge-disjoint and some customers may have already been satisfied) in $V(S)$ by doubling edges in $E(S)\cup \{e_S\}$ and shortcutting, where $e_S$ is a minimum weight edge between one customer in $V(S)$ and one depot in $O$.

The details are as follows.
First, for each customer $v\in V$ with $d(v)>k$, we construct $\ceil{\frac{d(v)}{k}}$ tours for $v$ by connecting $v$ with its nearest opened depot in $O$ and regard $v$ as a zero-demand customer in the following. 
Then, we consider a tree $T_u\in\T$ rooted at $u\in O$, and satisfy all non-zero-demand customers in $V(T_u)$ by splitting $T_u$. Denote the sub-tree rooted at $v$ and the set of $v$'s children by $T_v$ and $Q_v$, respectively.
\begin{itemize}
\item If $d(T_u)\leq k$, we construct a tour for all non-zero-demand customers in $V(T_u)$ by doubling all edges in $E(T_u)$ and then shortcutting. Note that $d(T_u)$ may have a demand of less than $k/2$, but $T_u$ has a good property: it contains an opened depot.
\item Else, we find a minimal sub-tree $T_v$ such that $d(T_v)>k$ and $d(T_{v'})\leq k$ for every $v'\in Q_v$. Consider sub-trees in $\T_v\coloneqq\{v\}\cup\{T_{v'}\mid v'\in Q_v\}$. We greedily partition them into sets $\T_0,\dots,\T_l$ such that $d(\T_i)\leq k$ for each $i$ and $d(\T_i)>k/2$ for each $i>0$.
For each $\T_i$ with $i>0$, we combine trees in $\T_i$ into a sub-tree $S$ by adding the edges joining $v$ and each tree in $\T_i$ and $v$ (if $\{v\}\notin \T_i$).
Then, we find a minimized cost edge $e_S$ connecting one depot in $O$ with one vertex in $V(S)\cup U(S)$. By doubling edges in $E(S)\cup\{e_S\}$ and shortcutting, we construct a tour for all non-zero-demand customers in $V(\T_i)$. At last, we update $T_u$ by removing $V(S)\setminus\{v\}$ and $E(S)$ from $T_u$, and regard $v$ as a zero-demand customer in the following if $\{v\}\in \T_i$.
\end{itemize}

%a set of customers $V'$ based on the structure property of trees in $\T$ with a total demand of at least $k/2$ and at most $k$, and then construct a tour for all customers in $V'$ by connecting a minimum cost edge between one customer in $V'$ and one depot in $O$.

The tree-splitting procedure used in Algorithm~\ref{algo:unsplittable UFL} is shown in Algorithm~\ref{algo:tree-splitting}.
The tours have the following properties.

\begin{lemma}[Lemma 3 in \cite{HarksKM13}]\label{lb3}
Given a set of opened depots $O$ and a constrained spanning forest $\T$, the set of tours $\I$ returned by the tree-splitting procedure holds that $(O,\I)$ is a feasible solution for unsplittable and splittable UFL, and 
$w(\I)\leq 2w(\T)+\sum_{v\in V}(4/k)d(v)\cdot\min_{u\in O}w(v,u)$.
\end{lemma}

\subsection{The Analysis}
For UFL with any constant $\gamma\geq 1.678$, Byrka and Aardal~\cite{DBLP:journals/siamcomp/ByrkaA10} proposed a (bifactor) $(1+2e^{-\gamma}, \gamma)$-approximation algorithm. It returns a solution whose connection cost is at most $(1+2e^{-\gamma})\cdot\psi_{LP}$ and whose opening cost is at most $\gamma\cdot\phi_{LP}$, where $\psi_{LP}$ and $\phi_{LP}$ are the values of the connection cost and the opening cost of an initially computed optimal
fractional LP solution, respectively. In Step~\ref{bifactor} of Algorithm~\ref{algo:unsplittable UFL}, we apply Byrka and Aardal's algorithm on the UFL instance to open a set of depots $O_1$. Hence, the connection cost and the opening cost in the solution satisfies $\sum_{v\in V}d(v)\cdot\min_{u\in O_1}\widetilde{w}(v,u)$ and $\widetilde{\phi}(O_1)$, respectively. Therefore, we have the following bounds.

\begin{lemma}\label{lb4}
It holds that $\sum_{v\in V}d(v)\cdot\min_{u\in O_1}\widetilde{w}(v,u)\leq (1+2e^{-\gamma})\cdot\psi_{LP}$ and $\widetilde{\phi}(O_1)\leq \gamma\cdot\phi_{LP}$.
\end{lemma}

Recall that $\OPT'$ is the cost of an optimal solution on the UFL instance in Lemma~\ref{lb1}. The cost of the optimal fractional LP solution is at most the cost of an optimal solution. Hence, we have $\psi_{LP}+\phi_{LP}\leq\OPT'$.
By Lemma~\ref{lb1}, we have the following lemma.

\begin{lemma}\label{lb5}
It holds that $\psi_{LP}+\phi_{LP}\leq\OPT'\leq\psi^*+\alpha\cdot\phi^*$.
\end{lemma}

\begin{algorithm}[H]
\caption{The tree-splitting procedure for CLR}
\label{algo:tree-splitting}
\small
\vspace*{2mm}
\textbf{Input:} An instance of CLR, a set of opened depots $O$, and a constrained spanning forest $\T$. \\
\textbf{Output:} A set of feasible tours $\I$ to CLR.

\begin{algorithmic}[1]
\State Initialize $\I=\emptyset$.

\For{all $v\in V$ with $d(v)> k$}

\State Construct $\ceil{\frac{d(v)}{k}}$ tours for $v$ by connecting $v$ with its nearest opened depot in $O$.
\State Add the tours to $\I$ and regard $v$ as a zero-demand customer in the following.
\EndFor

\For{all $T_u\in\T$}
\While{$d(T_u)>k$}\label{loop}
\State Find $v\in V(T_u)$ with $d(T_v)>k$ and $d(T_{v'})\leq k$ for each $v'\in Q_v$.\label{startloop}
\State Greedily partition trees in $\T_v\coloneqq\{v\}\cup\{T_{v'}\mid v'\in Q_v\}$ into sets $\T_0,\dots,\T_l$ such that $d(\T_i)\leq k$ for each $i$ and $d(\T_i)>k/2$ for each $i>0$. 
\For{$i\in\{1,\dots,l\}$}
\State Combine trees in $\T_i$ into a sub-tree $S$ by adding the edges joining $v$ and each tree in $\T_i$, and $v$ (if $\{v\}\notin \T_i$).\label{start-tour2}
\State Find an edge $e_S$ with minimized cost connecting one depot in $O$ with one vertex in $V(S)\cup U(S)$.
\State Construct a tour for all non-zero-demand customers in $V(\T_i)$ by doubling edges in $E(S)\cup \{e_S\}$ and shortcutting.
\State Add the tour to $\I$, update $T_u$ by removing $V(S)\setminus\{v\}$ and $E(S)$ from $T_u$, and regard $v$ as a zero-demand customer in the following if $\{v\}\in \T_i$.
\EndFor
\EndWhile
\State Construct a tour for all non-zero-demand customers in $V(T_u)$ by doubling edges in $E(T_u)$ and shortcutting.\label{tour3}
\State Add the tour to $\I$.
\EndFor
\State Return $\I$.
\end{algorithmic}
\end{algorithm}

\begin{theorem}\label{theorem1}
For unsplittable and splittable CLR, Tree-Alg is a polynomial-time $4.169$-approximation algorithm.
\end{theorem}
\begin{proof}
By Algorithm~\ref{algo:unsplittable UFL}, it returns a solution $(O,\I)$ such that %$\phi(O)=\phi(O_1)+\phi(O_2)$ and $w(\I)\leq 2w(\T)+\sum_{v\in V}(4/k)d(v)\cdot\min_{u\in O}w(v,u)$. Note that $\min_{u\in O}w(v,u)\leq \min_{u\in O_1}w(v,u)$, $\widetilde{w}=(2/k)w$, $\widetilde{\phi}=\alpha\cdot\phi$, and $\phi(O_2)=\psi(\T)$. Hence, the total cost is bounded by $2w(\T)+\sum_{v\in V}2d(v)\cdot\min_{u\in O_1}\widetilde{w}(v,u)+(1/\alpha)\cdot\widetilde{\phi}(O_1)+\phi(\T)$.
\begin{align*}
w(\I)+\phi(O)&\leq2w(\T)+\sum_{v\in V}(4/k)d(v)\cdot\min_{u\in O}w(v,u)+\phi(O_1)+\phi(O_2)\\
&\leq 2w(\T)+\sum_{v\in V}2d(v)\cdot\min_{u\in O_1}\widetilde{w}(v,u)+1/\alpha\cdot\widetilde{\phi}(O_1)+\phi(\T)\\
&\leq 2\psi^*+\phi^*+2(1+2e^{-\gamma})\cdot\psi_{LP}+(1/\alpha)\cdot\gamma\cdot\phi_{LP},
\end{align*}
where the first inequality follows from $w(\I)\leq 2w(\T)+\sum_{v\in V}(4/k)d(v)\cdot\min_{u\in O}w(v,u)$ by Lemma~\ref{lb3} and $\phi(O)\leq\phi(O_1)+\phi(O_2)$, the second inequality follows from $\min_{u\in O}w(v,u)\leq \min_{u\in O_1}w(v,u)$, $\widetilde{w}=(2/k)w$, $\widetilde{\phi}=\alpha\cdot\phi$, and $\phi(O_2)=\phi(\T)$, and the last inequality follows from $\OPT=\psi^*+\phi^*$ and Lemmas~\ref{lb2} and \ref{lb4}.

Let $f_\gamma(\alpha)\coloneqq \max\{2(1+2e^{-\gamma}),(1/\alpha)\cdot\gamma\}$. %Recall that $\OPT=\psi^*+\phi^*$. 
Since $\OPT=\psi^*+\phi^*$, the approximation ratio is 
\begin{align*}
\frac{2\psi^*+\phi^*+f_\gamma(\alpha)\cdot(\psi_{LP}+\phi_{LP})}{\psi^*+\phi^*}&\leq \frac{2\psi^*+\phi^*+f_\gamma(\alpha)\cdot(\psi^*+\alpha\cdot\phi^*)}{\psi^*+\phi^*}\\
&\leq\max\{2+f_\gamma(\alpha), 1+\alpha\cdot f_\gamma(\alpha)\},
\end{align*}
where the first inequality follows from Lemma~\ref{lb5}.

Setting $\alpha=1.461$ and $\gamma=3.168$, we get $f_\gamma(\alpha)\leq 2.169$ and $\max\{2+f_\gamma(\alpha), 1+\alpha\cdot f_\gamma(\alpha)\}\leq 4.169$. Hence, the approximation ratio of Algorithm~\ref{algo:unsplittable UFL} is at most $4.169$.

Since the bifactor approximation algorithm~\cite{DBLP:journals/siamcomp/ByrkaA10} and the tree-splitting procedure~\cite{HarksKM13} run in polynomial time, Tree-Alg also runs in polynomial time.
\end{proof}

\section{Splittable CLR}\label{sec5}
In this section, we introduce our Path-Alg, which is a $4.091$-approximation algorithm for splittable CLR.

%Recall that the main idea of the cycle-splitting method proposed in~\cite{tight} is to obtain a Hamiltonian cycle by using a $\delta$-approximation algorithm for metric TSP (recall that $\delta\approx 1.5$), and then obtain tours based on splitting the cycle.
%Since we can transform a cycle into a path by deleting an edge from it, instead of splitting a cycle, we focus on splitting a path.
%Although it seems that there is no big difference between cycles and paths, they may lead to different approximation ratios. This is because the edges incident to depots may have additional costs, and we need to carefully control the number of edges incident to depots.

In Path-Alg, we focus on splitting paths. A path, as a special case of tree, has a simpler structure: for the case of splittable demand, one can greedily divide it into several segments of demand $k$ and possibly leave the last segment (containing the depot) with a demand of less than $k$. 
Recall that the tree-splitting procedure is mainly to obtain a sub-tree $S$ such that $k\geq d(S)>k/2$ repeatedly.
Therefore, based on a path-splitting procedure, we may reduce the connection part of the routing cost from $\sum_{v\in V}(4/k)d(v)\cdot\min_{u\in O}w(v,u)$ to $\sum_{v\in V}(2/k)d(v)\cdot\min_{u\in O}w(v,u)$. 
However, using this method does not straightforwardly lead to a better approximation ratio for splittable CLR, as it may involve using a $\delta$-approximation algorithm for metric TSP to obtain a set of paths, which is more expensive than the cost of the constrained spanning forest computed in Lemma~\ref{lb2}.
We will obtain a good approximation ratio based on this idea with novel analysis.

\subsection{The Algorithm}
Compared to Tree-Alg, Path-Alg has two main modifications: firstly, it computes a constrained spanning path-packing rather than a constrained spanning forest, and secondly, it constructs a set of tours through a path-splitting procedure instead of the tree-splitting procedure.

To compute a constrained spanning path-packing, we construct two new graphs $G'$ and $H$. Given $G=(V\cup U, E, w,\phi)$, we obtain a new graph $G'=(V\cup U, E, w',\phi)$ such that $w'(v,v')=w(v,v')$ for any $v,v'\in V$ and $w'(u,v)=w(u,v)+\theta\cdot\phi(u)$ for any $v\in V$ and $u\in U$, where $\theta$ is a constant defined later. Then, we obtain another graph $H=(V\cup\{r\}, F, c)$ such that $c(r,v)=\min_{u\in U}w'(u,v)$ and $c(v,v')=\min\{w'(v,v'), c(r,v)+c(r,v')\}$. One can also think that $H$ is obtained by contracting all depots in $U$ as a super-depot $r$ and then taking a metric closure of $G'$. Note that the edge weight functions in new graphs $G'$ and $H$ are still metric functions.

To obtain the desired constrained spanning path-packing in $G$, we may consider using an approximation algorithm for metric TSP. 

\begin{lemma}\label{pathpack}
Given a $\delta$-approximation algorithm for metric TSP, there is a polynomial-time algorithm to compute a good constrained spanning path-packing $\P$ in $G$.
\end{lemma}
\begin{proof}
By applying a $\delta$-approximation algorithm for metric TSP in $H$, we can obtain a Hamiltonian cycle $C$ in $H$. Note that $C$ corresponds to a subgraph of $G'$ where each vertex in $V$ has an even degree but the vertex in $U$ may has an odd degree. 
Therefore, we can obtain a constrained spanning cycle-packing $\C$ in $G'$ by shortcutting.
Note that the shortcutting can ensure that each depot is contained in at most one cycle or path in $\C$.

Then, we can transform $\C$ into a constrained spanning path-packing $\P$ by doing: (1) for each cycle $uv_1\dots v_iu\in\C$ we delete the edge with a smaller cost from $uv_1$ and $uv_i$, and (2) for each path $uv_1\dots v_iu'\in\C$ we delete the edge incident to the depot with a smaller opening cost from $uv_1$ and $u'v_i$. We can see that each depot is contained in at most one path in $\P$ (Recall Figure~\ref{fig:whole_figure}).
\end{proof}

While the constrained spanning path-packing computed in Lemma~\ref{pathpack} may not be a $\delta$-approximate constrained spanning path-packing, it will still be useful in our analysis. An illustration of the constrained spanning path-packing algorithm in Lemma~\ref{pathpack} can be found in Figure~\ref{ALG-1}.

\begin{figure}[ht]
    \centering
    \begin{subfigure}[b]{0.48\textwidth}
        \centering
        \includegraphics[width=0.95\textwidth]{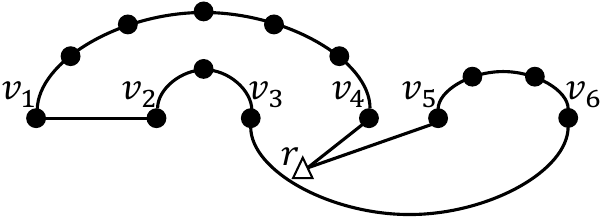}
        \caption{A $\delta$-approximate Hamiltonian cycle in graph $H$, where $H$ is obtained by contracting all depots in $G'$ as a super-depot $r$.\\}
        \label{bfig1}
    \end{subfigure}
    \hfill
    \begin{subfigure}[b]{0.48\textwidth}
        \centering
        \includegraphics[width=0.95\textwidth]{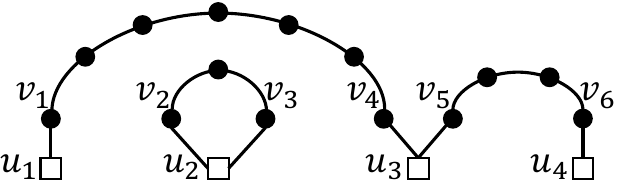}
        \caption{The corresponding subgraph of $G'$ w.r.t. the Hamiltonian cycle in $H$, where the edges $v_1v_2$, $v_3v_6$, $rv_4$, and $rv_5$ in $H$ correspond to the edges in $\{u_1v_1,u_2v_2\}$, $\{u_2v_3,u_4v_6\}$, $\{u_3v_4\}$, and $u_3v_5$.}
        \label{bfig2}
    \end{subfigure}
    
    \vspace{2mm}
    
    \begin{subfigure}[b]{0.48\textwidth}
        \centering
        \includegraphics[width=0.95\textwidth]{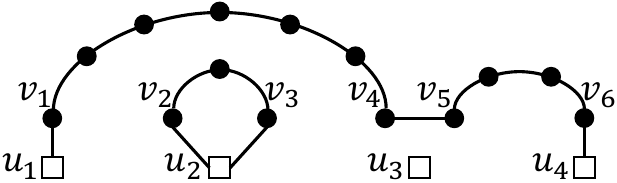}
        \caption{A constrained spanning cycle-packing $\C$ in $G'$, obtained by shortcutting the subgraph of $G'$.\\}
        \label{bfig3}
    \end{subfigure}
    \hfill
    \begin{subfigure}[b]{0.48\textwidth}
        \centering
        \includegraphics[width=0.95\textwidth]{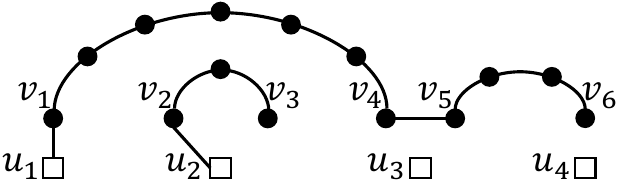}
        \caption{A constrained spanning path-packing $\P$ in $G'$ (and $G$), obtained by carefully deleting the edges incident to the depots in $\C$.}
        \label{bfig4}
    \end{subfigure}
    \caption{An illustration of the constrained spanning path-packing algorithm in Lemma~\ref{pathpack}, where each cycle node represents a customer, each square node represents a depot, and the triangle node represents the super-depot.}
    \label{ALG-1}
\end{figure}

The framework of our Path-Alg for splittable CLR can be seen in Algorithm~\ref{algo:splittable UFL}.

The path-splitting procedure works as follows. 
First, for each customer $v\in V$ with $d(v)>k$, we construct $\ceil{\frac{d(v)}{k}}-1$ tours for $v$ (with each delivering $k$ of demand) by connecting $v$ with its nearest opened depot in $O$ and update $d(v)\coloneqq d(v)-k\cdot (\ceil{\frac{d(v)}{k}}-1)$. In the following, the demand of each customer $v$ holds that $0<d(v)\leq k$. Then, we consider a path $P_u\in\P$ rooted at $u\in O$, and satisfy all customers in $V(P_u)$ by splitting $P_u$. Denote the sub-path rooted at $v$ and $v$'s unique children by $P_v$ and $q(v)$, respectively. 
We consider the following two cases.

\begin{itemize}
\item If $d(P_u)\leq k$, we get a tour for all customers in $V(P_u)$ by doubling all edges in $E(P_u)$ and then shortcutting.
\item Otherwise, we can do the following repeatedly until it satisfies that $d(P_u)\leq k$. First, we can find a customer $v\in V(P_u)$ such that $d(P_v)>k$ and $d(P_{q(v)})\leq k$. Consider the sub-path $P_v$. Then, we find an edge $e_{P_v}$ with minimized cost connecting one depot in $O$ with one vertex in $V(P_v)$. By doubling edges in $E(P_v)\cup\{e_S\}$ and shortcutting, we construct a tour $I$ for all customers in $V(P_v)$ with demand assignments: $x_{vI}=k-d(P_{q(v)})$ and $x_{v'I}=d(v')$ for each $v'\in V(P_{q(v)})$. At last, we update $P_u$ by removing $V(P_{q(v)})$ and $E(P_v)$ and setting $d(v)\coloneqq d(v)-x_{vI}$.
\end{itemize}

\begin{algorithm}[H]
\caption{An improved approximation algorithm for splittable CLR (Path-Alg)}
\label{algo:splittable UFL}
\small
\vspace*{2mm}
\textbf{Input:} An instance of CLR. \\
\textbf{Output:} A feasible solution to CLR.

\begin{algorithmic}[1]
\State Create an UFL instance with edge costs $\widetilde{w}=(2/k)w$ and depot costs $\widetilde{\phi}=\alpha\cdot\phi$ as in Lemma~\ref{lb1}.

\State Apply Byrka and Aardal's bifactor approximation algorithm~\cite{DBLP:journals/siamcomp/ByrkaA10} with a parameter of $\gamma>0$ on the UFL instance and let $O_1$ be the set of depots opened in the resulting UFL solution.

\State Compute a constrained spanning path-packing $\P$ in $G$ as in Lemma~\ref{pathpack} and let $O_2=U(\P)$ be the set of depots contained in some $P\in\P$.

\State Open all depots in $O\coloneqq O_1\cup O_2$.

\State Obtain a set of feasible tours $\I$ by calling the path-splitting procedure in Algorithm~\ref{algo:path-splitting}.

\State Return $(O, \I)$.
\end{algorithmic}
\end{algorithm}

The path-splitting procedure used in Algorithm~\ref{algo:splittable UFL} is shown in Algorithm~\ref{algo:path-splitting}. 
%We remark that the previous cycle-splitting procedure in \cite{tight} first computes a set of tours in $H$ (Recall that $H$ is obtained by contracting all depots in $U$ as a super-depot) by using the well-known cycle partition algorithm~\cite{HaimovichK85,altinkemer1987heuristics}, which corresponds to a spanning cycle-packing $\C$ in $G$. It then modifies $\C$ into a set of feasible tours by introducing additional costs.
An example of the path-splitting procedure can be found in Figure~\ref{ALG-2}.

\begin{algorithm}[H]
\caption{The path-splitting procedure for splittable CLR}
\label{algo:path-splitting}
\small
\vspace*{2mm}
\textbf{Input:} An instance of CLR, a set of opened depots $O$, and a constrained spanning path-packing $\P$. \\
\textbf{Output:} A set of feasible tours $\I$ to CLR.

\begin{algorithmic}[1]
\State Initialize $\I=\emptyset$.

\For{all $v\in V$ with $d(v)>k$}
\State Construct $\ceil{\frac{d(v)}{k}}-1$ tours for $v$ by connecting $v$ with its nearest opened depot in $O$ (each delivers $k$ of demand).
\State Add the tours to $\I$ and update $d(v)\coloneqq d(v)-k\cdot (\ceil{\frac{d(v)}{k}}-1)$.\label{tour1+}
\EndFor

\For{all $P_u\in\P$}\Comment{$P_u$ is rooted at the depot $u\in O$}
\While{$d(P_u)>k$}\label{loop+}
\State Find $v\in V(P_u)$ with $d(P_{v})>k$ and $d(P_{q(v)})\leq k$.\Comment{$q(v)$ is $v$'s unique children}\label{startloop+}
\State Find an edge $e_{P_v}$ with minimized cost connecting one depot in $O$ with one vertex in $V(P_v)$.
\State Construct a tour $I$ for all customers in $V(P_v)$ by doubling edges in $E(P_v)\cup \{e_{P_v}\}$ and shortcutting with demand assignments: $x_{vI}=k-d(P_{q(v)})$ and $x_{v'I}=d(v')$ for each $v'\in V(P_{q(v)})$.\label{tour2-}
\State Add the tour to $\I$, update $P_u$ by removing $V(P_{q(v)})$ and $E(P_v)$ from $P_u$, and update $d(v)\coloneqq d(v)-x_{vI}$.\label{tour2+}
\EndWhile
\State Construct a tour for all customers in $V(P_u)$ by doubling edges in $E(P_u)$ and shortcutting. Add the tour to $\I$.\label{tour3+}
%\State Add the tour to $\I$.\label{tour3+}
\EndFor
\State Return $\I$.
\end{algorithmic}
\end{algorithm}

\begin{figure}[ht]
    \centering
    \begin{subfigure}[b]{0.96\textwidth}
        \centering
        \includegraphics[width=0.95\textwidth]{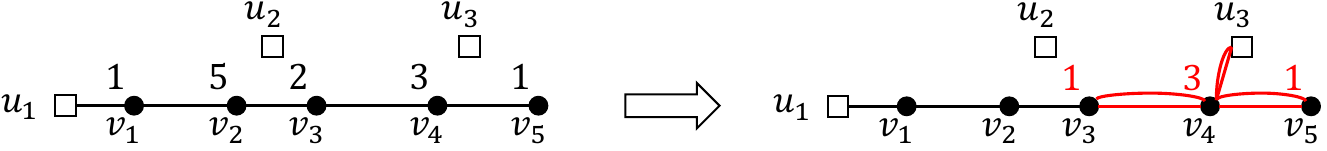}
        \caption{Consider a path $P_{u_1}=v_1...v_5$, rooted at $u_1$, in the constrained spanning path-packing $\P$, where the customers have a demand of 1, 5, 2, 3, and 1, respectively; In the first loop, the path-splitting procedure picks the sub-path $P_{v_3}$, finds an edge $e_{P_{v_3}}=u_3v_4$, and obtains a tour by doubling all edges in $E(P_{v_3})\cup \{e_{P_{v_3}}\}$ and shortcutting.}
        \label{cfig1}
    \end{subfigure}
    
    \vspace{2mm}

    \begin{subfigure}[b]{0.96\textwidth}
        \centering
        \includegraphics[width=0.95\textwidth]{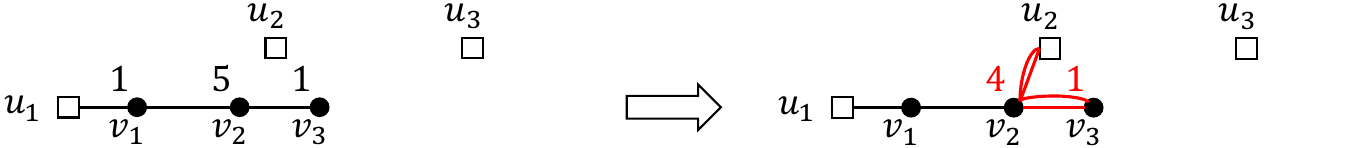}
        \caption{The path-splitting procedure updates $P_{u_1}$ by removing $V(P_{v_4})$ and $E(P_{v_3})$ from $P_{u_1}$ and updates $d(v_3)\coloneqq 2-1=1$, and then the remaining customers have a demand of 1, 5, and 1, respectively; In the second loop, the path-splitting procedure picks the sub-path $P_{v_2}$, finds an edge $e_{P_{v_3=2}}=u_2v_2$, and obtains a tour by doubling all edges in $E(P_{v_2})\cup \{e_{P_{v_2}}\}$ and shortcutting.}
        \label{cfig2}
    \end{subfigure}
    
    \vspace{2mm}
    
    \begin{subfigure}[b]{0.96\textwidth}
        \centering
        \includegraphics[width=0.95\textwidth]{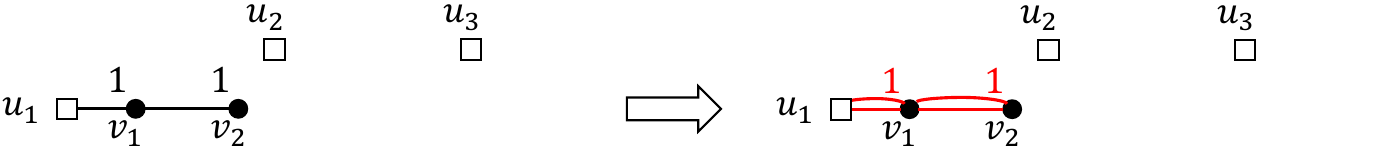}
        \caption{The path-splitting procedure updates $P_{u_1}$ by removing $V(P_{v_3})$ and $E(P_{v_2})$ from $P_{u_1}$ and updates $d(v_2)\coloneqq 5-4=1$, and then the remaining customers have a demand of 1, and 1, respectively; In the third loop, the path-splitting procedure picks the path $P_{u_1}$, obtains a tour by doubling all edges in $E(P_{u_1})$ and shortcutting, and removes $P_{u_1}$.}
        \label{cfig3}
    \end{subfigure}
    \caption{An example of the path-splitting procedure, where each cycle node represents a customer and each square node represents a depot.}
    \label{ALG-2}
\end{figure}

\begin{lemma}\label{lb6}
Given a set of opened depots $O$ and a constrained spanning path-packing $\P$, the path-splitting procedure can use polynomial time to obtain a set of tours $\I$ such that $(O,\I)$ is a feasible solution for splittable UFL and 
$w(\I)\leq 2w(\P)+\sum_{v\in V}(2/k)d(v)\cdot\min_{u\in O}w(v,u)$.
\end{lemma}
\begin{proof}
First, it is easy to see that each customer is satisfied, i.e., $\sum_{I\in\I}\sum_{v\in V(I)}x_{vI}=d(v)$.

Let $\I_1$, $\I_2$, and $\I_3$ be the set of tours collected in Steps~\ref{tour1+}, \ref{tour2+}, and \ref{tour3+}, respectively. Then, we have $\I=\I_1\cup\I_2\cup\I_3$. For a tour $I\in\I_1\cup\I_2$, we denote the edge connecting one customer in $I$ and one depot in $O$ by $e_I$. Then, since $\sum_{v\in V(I)}x_{vI}=k$ by Algorithm~\ref{algo:path-splitting}, we have 
$
w(e_I)=\min_{v\in V(I),u\in O}w(v,u)\leq\sum_{v\in V(I)}(1/k) x_{vI}\cdot\min_{u\in O}w(v,u).
$
%$w(e_I)=\min_{v\in V(I),u\in O}w(v,u)\leq\sum_{v\in V(I)}(1/k) x_{vI}\cdot\min_{u\in O}w(v,u)$ 
The tours in $\I$ are generated by doubling edges of $E(\P)\cup\{e_I\mid I\in\I_1\cup\I_2\}$ and shortcutting. Every edge in $E(\P)$ is used for only one tour since it will be removed in Step~\ref{tour2+} of Algorithm~\ref{algo:path-splitting} once after it is used in Step~\ref{tour2-}. 
Therefore, we can get that 
\begin{align*}
w(\I)&\leq2w(\P)+\sum_{I\in\I_1\cup\I_2}2w(e_I)\\
&\leq2w(\P)+\sum_{I\in\I_1\cup\I_2}\sum_{v\in V(I)}(2/k) x_{vI}\cdot\min_{u\in O}w(v,u)\\
&\leq2w(\P)+\sum_{I\in\I}\sum_{v\in V(I)}(2/k) x_{vI}\cdot\min_{u\in O}w(v,u)\\
&=2w(\P)+\sum_{v\in V}(2/k)d(v)\cdot\min_{u\in O}w(v,u),
\end{align*}
where
%the second inequality is from (\ref{eq2}) and 
the last equality from $\sum_{I\in\I}\sum_{v\in V(I)}x_{vI}=d(v)$.
\end{proof}

We remark that that Lemma~\ref{lb6} shows that almost all tours computed by Path-Alg deliver exactly $k$ of demand, whereas, by Lemma~\ref{lb3}, almost all tours computed by Tree-Alg deliver only $k/2$ of demand in the worst case.

\begin{comment}
\begin{proof}
First, it is easy to see that each customer is satisfied.
For the sake of analysis, we regard a customer $v$ as $d(v)$ customers with unit-demand at the same position and let $\widetilde{V}\coloneqq\bigcup\{v\}_{i=1}^{d(v)}$. So, we only need to prove $w(\I)\leq 2w(\P)+\sum_{v\in\widetilde{V}}(2/k)\cdot\min_{u\in O}w(v,u)$.

Let $\I_1$, $\I_2$, and $\I_3$ be the set of tours collected in Steps~\ref{tour1+}, \ref{tour2+}, and \ref{tour3+}, respectively. Then, we have $\I=\I_1\cup\I_2\cup\I_3$. For a tour $I\in\I_1\cup\I_2$, we denote the edge connecting one customer in $I$ and one depot in $O$ by $e_I$. Then, we have $w(e_I)=\min_{v\in \widetilde{V}(I),u\in O}w(v,u)\leq \sum_{v\in\widetilde{V}(I)}(1/k)\cdot\min_{u\in O}w(v,u)$ since $d(I)=\size{\widetilde{V}(I)}=k$. The tours in $\I$ are generated by doubling edges of $E(\P)\cup\{e_I\mid I\in\I_1\cup\I_2\}$ and shortcutting. Note that we ignore the zero weighted between customers at the same position. Every edge in $E(\P)$ is used for only one tour since it will be removed in Step~\ref{tour2+}. So, we have 
\begin{align*}
w(\I)\leq&\ 2w(\P)+\sum_{I\in\I_1\cup\I_2}2w(e_I)\\
\leq&\ 2w(\P)+\sum_{I\in\I_1\cup\I_2}\sum_{v\in \widetilde{V}(I)}(2/k)\cdot\min_{u\in O}w(v,u)\\
\leq&\ 2w(\P)+\sum_{I\in\I}\sum_{v\in \widetilde{V}(I)}(2/k)\cdot\min_{u\in O}w(v,u)\\
=&\ 2w(\P)+\sum_{v\in\widetilde{V}}(2/k)\cdot\min_{u\in O}w(v,u),
\end{align*}
where the last equality follows from $\widetilde{V}=\bigcup_{I\in\I}\widetilde{V}(I)$.
\end{proof}
\end{comment}

\subsection{The Analysis}
Based on a straightforward analysis, we may obtain $2w(\P)+\phi(\P)\leq 2\delta\cdot\psi^*+2\delta\cdot\phi^*$, where $\delta$ is the approximation ratio of the algorithm used for metric TSP. However, it cannot lead to a better-than-4.169-approximation ratio. In this subsection, we will prove a better bound based on a tighter analysis, and obtain a $4.091$-approximation ratio.

Let $C^*$ be a minimum cost Hamiltonian cycle in graph $H$. 

\begin{lemma}\label{lb8}
It holds that $c(C^*)\leq \psi^*+2\theta\cdot\phi^*$.
\end{lemma}
\begin{proof}
Consider an optimal solution $(O,\I)$ of CLR. For each depot $o\in O$, there is a set of tours $\I_o\subseteq\I$ with each in it containing $o$. By shortcutting all tours in $\I_o$, we obtain a cycle $C_o=ov_1\dots v_io$ such that $w'(C_o)\leq w(\I_o)+2\theta\cdot\phi(o)$. Recall that the graph $H$ is obtained by contracting all depots in $U$ as a super-depot $r$ and then taking a metric closure. Thus, $\{C_o\}_{o\in O}$ corresponds to an Eulerian graph in $H$. Therefore, by shortcutting the corresponding Eulerian graph, we obtain a Hamiltonian cycle in $H$ with a cost of at most $\sum_{o\in O}w'(C_o)$. Since $C^*$ is the minimum cost Hamiltonian cycle in $H$, we can get that 
$
c(C^*)\leq \sum_{o\in O}w'(C_o)\leq w(\I)+2\theta\cdot\phi(O)=\psi^*+2\theta\cdot\phi^*,   
$
which finishes the proof.
%$c(C^*)\leq \sum_{o\in O}w'(C_o)\leq w(\I)+2\theta\cdot\phi(O)=\psi^*+2\theta\cdot\phi^*$.
\end{proof}

\begin{lemma}\label{lb9}
Given a $\delta$-approximation algorithm for metric TSP, the computed constrained spanning path-packing $\P$ in Lemma~\ref{pathpack} satisfies that $2w(\P)+\phi(\P)\leq 2\delta\cdot\psi^*+\delta\cdot\phi^*$ under $\theta=\frac{1}{4}$.
\end{lemma}
\begin{proof}
Given a $\delta$-approximate Hamiltonian cycle in $H$, the algorithm in Lemma~\ref{pathpack} simply computes a constrained spanning cycle-packing $\C$ in $G'$ by shortcutting the corresponding subgraph in $G'$. Hence, we have $w'(\C)\leq \delta\cdot c(C^*)$. Let $\C=\C_1\cup\C_2$, where $\C_1$ and $\C_2$ denote the set of cycles and paths in $\C$, respectively. 

For each cycle $uv_1\dots v_iu\in\C_1$ (resp., path $uv_1\dots v_iu'\in\C_2$), the algorithm deletes the edge with a smaller cost from $uv_1$ and $uv_i$ (resp., incident to the depot with a smaller opening cost from $uv_1$ and $u'v_i$). Therefore, the constrained spanning path-packing $\P$ satisfies that $w'(\P)\leq w'(\C)-\theta\cdot\phi(\C_1)-\frac{1}{2}\theta\cdot\phi(\C_2)$ and $\phi(\P)\leq \phi(\C_1)+\frac{1}{2}\phi(\C_2)$. Note that $w'(\P)=w(\P)+\theta\cdot\phi(\P)$ since each path in $\P$ contains only one edge incident to the depot. Moreover, under $\theta=\frac{1}{4}$, we also have $1-2\theta=2\theta>0$. Hence, we have 
\begin{align*}
2w(\P)+\phi(\P)&= 2w'(\P)+(1-2\theta)\cdot\phi(\P)\\
&\leq 2w'(\P)+(1-2\theta)\cdot\lrA{\phi(\C_1)+\frac{1}{2}\phi(\C_2)}\\
&=2w'(\P)+2\theta\cdot\lrA{\phi(\C_1)+\frac{1}{2}\phi(\C_2)}\\
&\leq 2w'(\C)\leq 2\delta\cdot c(C^*).
\end{align*}
So, we get $2w(\P)+\phi(\P)\leq 2\delta\cdot\psi^*+\delta\cdot\phi^*$ by Lemma~\ref{lb8}.
\end{proof}

Using the well-known Christofides-Serdyukov algorithm~\cite{christofides1976worst,serdyukov1978some} for metric TSP, we have $\delta=\frac{3}{2}$, and then we have $2w(\P)+\phi(\P)\leq 3\psi^*+\frac{3}{2}\phi^*$ by Lemma~\ref{lb9}. Note that this result may only lead to a 4.143-approximation algorithm for splittable CLR. We can obtain a further improvement if using the property of their algorithm. Let $T^*$ be a minimum spanning tree in $H$. We have the following lemma.

\begin{lemma}\label{lb7}
It holds that $c(T^*)\leq\psi^*+\theta\cdot\phi^*$.
\end{lemma}
\begin{proof}
%It follows a similar argument from the proof of Lemma~\ref{lb2}.
Consider an optimal solution $(O,\I)$ of CLR. For each depot $o\in O$, there is a set of tours $\I_o\subseteq\I$ with each in it containing $o$. By shortcutting all tours in $\I_o$ and deleting an edge incident to $o$, we obtain a path $P_o=ov_1\dots v_i$ such that $w'(P_o)\leq w(\I_o)+\theta\cdot\phi(o)$. Note that $\{P_o\}_{o\in O}$ corresponds to a graph spanning all vertices in $H$ with a cost of at most $\sum_{o\in O}w'(P_o)\leq w(\I)+\theta\cdot\phi(O)=\psi^*+\theta\cdot\phi^*$. Hence, we have $c(T^*)\leq\psi^*+\theta\cdot\phi^*$.
\end{proof}

\begin{lemma}\label{lb10}
Given the Christofides-Serdyukov algorithm for metric TSP, the computed constrained spanning path-packing $\P$ satisfies that $2w(\P)+\phi(\P)\leq 3\psi^*+\phi^*$ if $\theta=\frac{1}{4}$.
\end{lemma}
\begin{proof}
The Christofides-Serdyukov algorithm first computes a minimum spanning tree $T^*$ in $H$ and then a minimum cost matching $M$ on the set of vertices of $T^*$ with odd-degrees. It is well-known that $c(M)\leq \frac{1}{2}c(C^*)$. So, the constrained spanning cycle-packing $\C$ in $G'$ obtained by the algorithm in Lemma~\ref{pathpack} satisfies that $w'(\C)\leq c(T^*)+c(M)\leq c(T^*)+\frac{1}{2}c(C^*)\leq \psi^*+\theta\cdot\phi^*+\frac{1}{2}(\psi^*+2\theta\cdot\phi^*)=\frac{3}{2}\psi^*+2\theta\cdot\phi^*$ by Lemmas~\ref{lb8} and~\ref{lb10}. Under $\theta=\frac{1}{4}$, we have $2w(\P)+\phi(\P)\leq 2w'(\C)\leq 3\psi^*+\phi^*$ by the proof of Lemma~\ref{lb9}.
\end{proof}

\begin{theorem}\label{theorem2}
For splittable CLR, Path-Alg is a polynomial-time $4.091$-approximation algorithm.
\end{theorem}
\begin{proof}
By Algorithm~\ref{algo:splittable UFL}, it returns a solution $(O,\I)$ such that
\begin{align*}
w(\I)+\phi(O)&\leq 2w(\P)+\sum_{v\in V}(2/k)d(v)\cdot\min_{u\in O}w(v,u)+\phi(O_1)+\phi(O_2)\\
&\leq 2w(\P)+\sum_{v\in V}d(v)\cdot\min_{u\in O_1}\widetilde{w}(v,u)+(1/\alpha)\cdot\widetilde{\phi}(O_1)+\phi(\P)\\
&\leq 3\psi^*+\phi^*+(1+2e^{-\gamma})\cdot\psi_{LP}+(1/\alpha)\cdot\gamma\cdot\phi_{LP},
\end{align*}
where the first inequality follows from $w(\I)\leq 2w(\P)+\sum_{v\in V}(2/k)d(v)\cdot\min_{u\in O}w(v,u)$ by Lemma~\ref{lb6} and $\phi(O)\leq\phi(O_1)+\phi(O_2)$, the second inequality follows from $\min_{u\in O}w(v,u)\leq \min_{u\in O_1}w(v,u)$, $\widetilde{w}=(2/k)w$, $\widetilde{\phi}=\alpha\cdot\phi$, and $\phi(O_2)=\phi(\P)$, and the last inequality follows from $\OPT=\psi^*+\phi^*$ and Lemmas~\ref{lb4} and~\ref{lb9}.

Let $f_\gamma(\alpha)\coloneqq \max\{(1+2e^{-\gamma}),(1/\alpha)\cdot\gamma\}$. %Recall that $\OPT=\psi^*+\phi^*$. 
Since $\OPT=\psi^*+\phi^*$, the approximation ratio is at most
\begin{align*}
\frac{3\psi^*+\phi^*+f_\gamma(\alpha)\cdot(\psi_{LP}+\phi_{LP})}{\psi^*+\phi^*}&\leq \frac{3\psi^*+\phi^*+f_\gamma(\alpha)\cdot(\psi^*+\alpha\cdot\phi^*)}{\psi^*+\phi^*}\\
&\leq\max\{3+f_\gamma(\alpha), 1+\alpha\cdot f_\gamma(\alpha)\}, 
\end{align*}
where the first inequality follows from Lemma~\ref{lb5}.

Setting $\alpha=2.8332$ and $\gamma=3.0909$, we get $f_\gamma(\alpha)\leq 1.09096$ and $\max\{2+f_\gamma(\alpha), 1+\alpha\cdot f_\gamma(\alpha)\}\leq 4.091$. Hence, the approximation ratio of Algorithm~\ref{algo:splittable UFL} is at most $4.091$.

Moreover, it is easy to see that Path-Alg also runs in polynomial time.
\end{proof}

Recall that the main advantage of Path-Alg is that almost all tours it computes deliver exactly $k$ of demand. This advantage primarily arises from the properties of the splittable case. However, in the unsplittable case, maintaining this advantage becomes challenging, which in turn makes it difficult to achieve a better approximation ratio than Tree-Alg.

\section{Experimental Results}\label{sec6}
We conduct experiments to compare our two algorithms, Tree-Alg and Path-Alg, with the previous approximation algorithm in~\cite{HarksKM13}.
%Previously, we proposed two better approximation algorithms for CLR. In this section, we consider the experimental study: 
Next, we introduce the benchmark instances of CLR, the implementations of our algorithms, and the results, respectively.

\subsection{Instances}
Harks \emph{et al.}~\cite{HarksKM13} tested their approximation algorithm on a total of 45 CLR benchmark instances, including 36 instances from~\cite{tuzun1999two} and 9 instances from~\cite{barreto2007using}. Note that these instances were primarily considered in the unsplittable case.
They compared their results with the \emph{previous best known solutions} ({pbks}) obtained by heuristic approaches~\cite{prins2007solving,baldacci2009capacitated,barreto2007using,tuzun1999two} by computing the gaps between their results and the {pbks}, where each gap is calculated as $Gap=\frac{Result-pbks}{pbks}$. 
Notably, they also employed the LKH heuristic algorithm~\cite{helsgaun2000effective} for metric TSP to optimize the tours derived from their approximation algorithm.
Although some results of these benchmark instances have been slightly improved further~\cite{baldacci2011exact, contardo2014exact}, for the sake of comparison, we test our algorithms on these instances and still compute the gaps between our results and the {pbks}.
The details of the 45 tested instances, including the number of depots $m$, the number of customers $n$, the vehicle capacity $k$, the pbks, the gaps of the previous approximation algorithm (pGap), and the gaps of the previous approximation algorithm using the LKH algorithm (pGap+LKH) are shown in Table~\ref{insinfo}. 
%Therefore, we only need to compare our gaps and the gaps in~\cite{HarksKM13}.\footnote{\citeauthor{HarksKM13}~[\citeyear{HarksKM13}] also created some large random instances and tested their algorithm on these instances. However, these instances are not available on their provided website now, and hence we do not consider these instances.} 
%, where the columns represent the names of instances, the number of depots $m$, the number of customers $n$, the vehicle capacity $k$, the previous best-known solutions (pbks), the gaps of the previous approximation algorithm (pGap), and the last line calculates the average gaps.

\subsection{Implementations}
We present the detailed implementations of our algorithms.

On one hand, given a sub-tree, instead of finding tours by the method of doubling and shortcutting, we obtained a tour by finding a minimum cost matching on the odd-degree vertices contained in the sub-tree and then shortcutting. 
This was motivated by its ability to guarantee a $1.5$-approximation for metric TSP~\cite{christofides1976worst,serdyukov1978some}, whereas the doubling-and-shortcutting method may only achieve a 2-approximation~\cite{williamson2011design}. 
Note that in our Path-Alg we used this method to implement the 1.5-approximation algorithm of metric TSP as well.
The above method is based on the 1.5-approximation algorithm for metric TSP.
Moreover, since the LKH heuristic algorithm~\cite{helsgaun2000effective} can compute very good solutions for metric TSP on small instances very quickly, we alternatively used the LKH algorithm to find tours for Tree-Alg and Path-Alg, as well as the Hamiltonian cycle for Path-Alg. The results obtained by these two approaches will both be reported.

\begin{table}[H]
\centering
\resizebox{0.5\textheight}{!}{
\begin{tabular}{|l|c|c|c|c|c|c|c|c|c|c|c|c|}
\hline 
{\textbf{Names}} & {\boldmath \textbf{$m$}} & {\boldmath \textbf{$n$}} & {\boldmath \textbf{$k$}} & {\textbf{pbks}} & {\textbf{pGap}} & {\textbf{pGap+LKH}}\\
\hline
Chr69-100$\times$10 & 10 & 100 & 200 & 842.90 & 0.283 & 0.108\\
Chr69-50$\times$5 & 5 & 50 & 160 & 565.60 & 0.220 & 0.079\\
Chr69-75$\times$10 & 10 & 75 & 160 & 861.60 & 0.177&0.104\\
Gas67-22$\times$5 & 5 & 22 & 4500 & 585.11 & 0.244&0.201\\
Gas67-29$\times$5 & 5 & 29 & 4500 & 512.10 & 0.279&0.165\\
Gas67-32$\times$5 & 5 & 32 & 8000 & 562.20 & 0.245&0.179\\
Gas67-32$\times$5-2 & 5 & 32 & 11000 & 504.30 & 0.205&0.123\\
Gas67-36$\times$5 & 5 & 36 & 250 & 460.40 & 0.448&0.094\\
Min92-27$\times$5 & 5 & 27 & 2500 & 3062.00 & 0.181&0.115\\
P111112 & 10 & 100 & 150 & 1468.40 & 0.207&0.079\\
P111122 & 20 & 100 & 150 & 1449.20 & 0.235&0.117\\
P111212 & 10 & 100 & 150 & 1396.46 & 0.133&0.043\\
P111222 & 20 & 100 & 150 & 1432.29 & 0.342&0.246\\
P112112 & 10 & 100 & 150 & 1167.53 & 0.164&0.076\\
P112122 & 20 & 100 & 150 & 1102.70 & 0.133&0.095\\
P112212 & 10 & 100 & 150 & 793.97 & 0.086&0.041\\
P112222 & 20 & 100 & 150 & 728.30 & 0.119&0.070\\
P113112 & 10 & 100 & 150 & 1238.49 & 0.183&0.090\\
P113122 & 20 & 100 & 150 & 1246.34 & 0.201&0.131\\
P113212 & 10 & 100 & 150 & 902.38 & 0.140&0.082\\
P113222 & 20 & 100 & 150 & 1021.31 & 0.166&0.126\\
P121112 & 10 & 200 & 150 & 2281.78 & 0.217&0.142\\
P121122 & 20 & 200 & 150 & 2185.55 & 0.139&0.110\\
P121212 & 10 & 200 & 150 & 2234.78 & 0.191&0.067\\
P121222 & 20 & 200 & 150 & 2259.52 & 0.225&0.102\\
P122112 & 10 & 200 & 150 & 2101.90 & 0.145&0.081\\
P122122 & 20 & 200 & 150 & 1709.56 & 0.179&0.230\\
P122212 & 10 & 200 & 150 & 1467.54 & 0.107&0.050\\
P122222 & 20 & 200 & 150 & 1084.78 & 0.119&0.123\\
P123112 & 10 & 200 & 150 & 1973.28 & 0.170&0.098\\
P123122 & 20 & 200 & 150 & 1957.23 & 0.126&0.075\\
P123212 & 10 & 200 & 150 & 1771.06 & 0.183&0.068\\
P123222 & 20 & 200 & 150 & 1393.62 & 0.182&0.029\\
P131112 & 10 & 150 & 150 & 1866.75 & 0.253&0.145\\
P131122 & 20 & 150 & 150 & 1841.86 & 0.230&0.050\\
P131212 & 10 & 150 & 150 & 1981.37 & 0.153&0.105\\
P131222 & 20 & 150 & 150 & 1809.25 & 0.206&0.122\\
P132112 & 10 & 150 & 150 & 1448.27 & 0.163&0.088\\
P132122 & 20 & 150 & 150 & 1444.25 & 0.301&0.125\\
P132212 & 10 & 150 & 150 & 1206.73 & 0.101&0.050\\
P132222 & 20 & 150 & 150 & 931.94 & 0.170&0.049\\
P133112 & 10 & 150 & 150 & 1699.92 & 0.155&0.081\\
P133122 & 20 & 150 & 150 & 1401.82 & 0.127&0.050\\
P133212 & 10 & 150 & 150 & 1199.51 & 0.128&0.146\\
P133222 & 20 & 150 & 150 & 1152.86 & 0.081&0.134\\
\hline
\textbf{Average Gap} &&&&&0.188&0.100\\
\hline
\end{tabular}}
\caption{The information of the 45 CLR benchmark instances, where the last line calculates the average gaps.}
\label{insinfo}
\end{table}

On the other hand, as in~\cite{HarksKM13}, we used the practical 1.861-approximation algorithm~\cite{DBLP:journals/jacm/JainMMSV03} of UFL to open a set of depots $O_1$ instead of using the bifactor approximation algorithm~\cite{DBLP:journals/siamcomp/ByrkaA10} since the latter algorithm involves solving an LP, which is not practical. So, there was no need to consider the setting of $\gamma$. Moreover, after opening a set of depots using the greedy algorithm, suggested in~\cite{HarksKM13}, we regarded their opening cost as zero in the following. This does not impact the approximation ratio. For $\alpha$, we tested different values. This choice was made because for a range of values of $\alpha$ the implementations using the 1.5-approximation algorithm for metric TSP can always guarantee a good approximation ratio. See the following lemma. Harks \emph{et al.}~\cite{HarksKM13} only considered $\alpha=1$ and showed the implementation had a ratio of 5.722.
%We found that under $\alpha=1$ the greedy algorithm opened only one depot for each benchmark instance, and hence if we set $\alpha$ as a larger number it would still open the same depot. Therefore, we set $\alpha$ with $1\leq \alpha^{-1}\leq 50$ to test our algorithms.

In the following, we use ``Tree/Path-APP'' to denote that the implementation of Tree/Path-Alg is based on the 1.5-approximation algorithm for metric TSP, and use ``Tree/Path-LKH" to denote that the implementation of Tree/Path-Alg is based on the LKH algorithm for metric TSP.

\begin{lemma}\label{fact}
For any $0.5\leq \alpha\leq 1.26$, Tree-APP has an approximation ratio of 5.722; For any $1\leq \alpha\leq 2.07$, Path-APP has an approximation ratio of 4.861; Moreover, both implementations have a running-time bound of $O(n^3+nm\log nm)$.
\end{lemma}
\begin{proof}
The implementations of our algorithms have undergone two main changes. %, and we discuss each of them separately.

For the first change, it allows for the preservation of the upper bounds in both tree-splitting and path-splitting procedures, i.e., Lemmas~\ref{lb3} and \ref{lb6} still hold.
The reason is shown as follows.

When given a sub-tree $T$ (a sub-path is also considered as a sub-tree) in the procedures, we find a tour $I$ by the doubling-and-shorcutting method, which guarantees $w(I)\leq 2w(T)$. However, in the implementation, we obtain a tour $I'$ by finding a minimum cost matching $M$ on the odd-degree vertices in $U(T)\cup V(T)$ and then shortcutting, which still guarantees $w(I')\leq 2w(T)$ since $w(M)\leq w(T)$~\cite{christofides1976worst,serdyukov1978some}.

Next, we consider the impact of the second change.

Since we use the 1.861-approximation algorithm of UFL to open a set of depots $O_1$, we have 
$
\sum_{v\in V}d(v)\cdot\min_{u\in O_1}\widetilde{w}(v,u)+\widetilde{\phi}(O_1)\leq \rho\cdot\OPT'\leq \rho\cdot\psi^*+\alpha\cdot\rho\cdot\phi^*,
$
where $\rho=1.861$, and the second inequality follows from Lemma~\ref{lb5}.

By a similar proof of Theorem~\ref{theorem1}, we can get that Tree-APP has a weight of at most 
\begin{align*}
&2w(\T)+\sum_{v\in V}2d(v)\cdot\min_{u\in O_1}\widetilde{w}(v,u)+(1/\alpha)\cdot\widetilde{\phi}(O_1)+\phi(\T)\\
&\leq 2\psi^*+\phi^*+\max\{2,1/\alpha\}\lrA{\sum_{v\in V}d(v)\cdot\min_{u\in O_1}\widetilde{w}(v,u)+\widetilde{\phi}(O_1)}\\
&\leq 2\psi^*+\phi^*+\max\{2,1/\alpha\}\lrA{\rho\cdot\psi^*+\alpha\cdot\rho\cdot\phi^*}\quad\quad\quad\quad\quad\quad\\
%\end{align*}\begin{align*}
&\leq 2\psi^*+\phi^*+2\lrA{\rho\cdot\psi^*+\frac{2\rho+1}{2\rho}\cdot\rho\cdot\phi^*}=(2\rho+2)\cdot\OPT,
\end{align*}
where the last inequality follows from $0.5\leq\alpha\leq1.26<\frac{2\rho+1}{2\rho}$ and the last equality from $\OPT=\psi^*+\phi^*$. So, for any $0.5\leq\alpha\leq1.26$, Tree-APP always has an approximation ratio of $2\rho+2=5.722$.

By a similar proof of Theorem~\ref{theorem2}, we can get that Path-APP has a weight of at most 
\begin{align*}
&2w(\P)+\sum_{v\in V}d(v)\cdot\min_{u\in O_1}\widetilde{w}(v,u)+(1/\alpha)\cdot\widetilde{\phi}(O_1)+\phi(\P)\\
&\leq 3\psi^*+\phi^*+\max\{1,1/\alpha\}\lrA{\sum_{v\in V}d(v)\cdot\min_{u\in O_1}\widetilde{w}(v,u)+\widetilde{\phi}(O_1)}\\
&\leq 3\psi^*+\phi^*+\max\{1,1/\alpha\}\lrA{\rho\cdot\psi^*+\alpha\cdot\rho\cdot\phi^*}\\
&\leq 3\psi^*+\phi^*+\lrA{\rho\cdot\psi^*+\frac{\rho+2}{\rho}\cdot\rho\cdot\phi^*}=(\rho+3)\cdot\OPT,
\end{align*}
where the last inequality follows from $1\leq\alpha\leq2.07<\frac{\rho+2}{\rho}$ and the last equality follows from $\OPT=\psi^*+\phi^*$. Hence, for any $1\leq\alpha\leq2.07$, Path-APP has an approximation ratio of $\rho+3=4.861$.

Last, we show the running time of our implementations.

The running time of our implementations is dominated by computing minimum cost matchings to obtain tours and using the greedy algorithm of UFL to open depots. 
Given a sub-tree with $s$ vertices, it takes $O(s^3)$ time to find a minimum cost matching on the odd-degree vertices~\cite{christofides1976worst,serdyukov1978some}. The total number of vertices in the sub-trees in our algorithms is bounded by $O(n)$; hence, the running time of finding minimum cost matchings is $O(n^3)$. 
The greedy algorithm of UFL can be implemented in $O(nm\log nm)$ time~\cite{HarksKM13}.
Therefore, our implementations take $O(n^3+nm\log nm)$ time.
%Our second change is from the implementation in~\cite{HarksKM13}. Since their implementation takes $O(n^2m)$ time, and given a sub-tree on $n$ vertices it requires $O(n^3)$ time to find a minimum cost matching on the odd-degree vertices and shortcutting~\cite{christofides1976worst,serdyukov1978some}, it is easy to see that our implementations take $O(n^3+n^2m)$ time.
%Note that the above analysis shows that for a specific range of $\alpha$ the implementations of Tree-Alg and Path-Alg always have a good approximation ratio. For $\alpha$ outside the range, the ratio may become worse. However, in practical experiments, the performance might improve, as approximation algorithms focus on worst-case scenarios.
%Note that the above analysis shows that for any $0.5\leq\alpha\leq 1.26$ the implementation of Tree-Alg has an approximation ratio of $5.722$ and for other values of $\alpha$ the ratio may become worse
\end{proof}

%The previous implementation of the algorithm in~\cite{HarksKM13} also achieves a ratio of 5.722.
Our algorithms are implemented in C++ on a desktop computer with an AMD Ryzen 5 PRO 4650G with Radeon Graphics (3.70 GHz, 32.0 GB RAM) using Windows Subsystem for Linux (WSL).
The detailed information of our algorithms, and the 45 tested instances can be found in \url{http://github.com/zhmjsnl/CLR}.
%Lastly, since the biggest benchmark instance contains only 220 vertices in total, we can quickly find an optimal solution of TSP by an ILP using Gurobi~\cite{gurobi}. Hence, whenever our algorithms needed to find a tour visiting a set of vertices, we used the ILP to find the best tour.

%If we set $\alpha=1$, the implementation of our first algorithm is the same as in~\cite{HarksKM13}.
%Therefore, our first algorithm should performance similarly to their algorithm.

\begin{table}[H]
\centering
\resizebox{0.9\textwidth}{!}{
\begin{tabular}{|c|c|c|c|c|c|c|c|c|}
\hline
\multirow{2}*{\boldmath $\alpha$} & \multicolumn{2}{c|}{\textbf{Tree-APP}} & \multicolumn{2}{c|}{\textbf{Tree-LKH}} & \multicolumn{2}{c|}{\textbf{Path-APP}} & \multicolumn{2}{c|}{\textbf{Path-LKH}} \\
\cline{2-9}
& \textbf{Gap} & \textbf{Time} & \textbf{Gap} & \textbf{Time} & \textbf{Gap} & \textbf{Time} & \textbf{Gap} & \textbf{Time}\\
\hline
0.1 & 0.1996 & 0.080 & 0.1780 & 0.313 & 0.2621 & 0.082 & 0.2009 & 1.321\\
0.2 & 0.1503 & 0.078 & 0.1295 & 0.318 & 0.2098 & 0.081 & 0.1455 & 1.180\\
0.3 & 0.1265 & 0.077 & 0.1052 & 0.348 & 0.1883 & 0.080 & 0.1253 & 1.194\\
0.4 & $\boldsymbol{0.1164}$ & 0.077 & $\boldsymbol{0.0944}$ & 0.318 & 0.1755 & 0.080 & 0.1166 & 1.291\\
0.5 & 0.1211 & 0.077 & 0.0984 & 0.356 & 0.1755 & 0.081 & 0.1134 & 1.303\\
0.6 & 0.1202 & 0.077 & 0.0980 & 0.327 & 0.1677 & 0.083 & 0.1084 & 1.361\\
0.7 & 0.1219 & 0.076 & 0.0985 & 0.345 & $\boldsymbol{0.1676}$ & 0.082 & $\boldsymbol{0.1077}$ & 1.226\\
0.8 & 0.1270 & 0.077 & 0.1036 & 0.327 & 0.1733 & 0.082 & 0.1099 & 1.242\\
0.9 & 0.1330 & 0.076 & 0.1103 & 0.343 & 0.1834 & 0.082 & 0.1127 & 1.546\\
1.0 & 0.1363 & 0.076 & 0.1145 & 0.333 & 0.1874 & 0.082 & 0.1125 & 1.549\\
1.1 & 0.1408 & 0.077 & 0.1191 & 0.326 & 0.1872 & 0.083 & 0.1164 & 1.578\\
1.2 & 0.1420 & 0.077 & 0.1206 & 0.345 & 0.1888 & 0.083 & 0.1183 & 1.574\\
1.3 & 0.1480 & 0.077 & 0.1271 & 0.325 & 0.1932 & 0.083 & 0.1245 & 1.381\\
1.4 & 0.1566 & 0.077 & 0.1359 & 0.349 & 0.1998 & 0.083 & 0.1322 & 1.432\\
1.5 & 0.1589 & 0.076 & 0.1375 & 0.317 & 0.2024 & 0.083 & 0.1345 & 1.434\\
\hline 
\end{tabular}}
\caption{The average gaps between our results and the pbks for Tree-Alg and Path-Alg under different values of $\alpha$.
}\label{experiment}
\end{table}

\subsection{Results}
For each instance, the gap of our result to the pbks is also calculated by $Gap=\frac{Result-pbks}{pbks}$. 
A summary of our results can be found in Table~\ref{experiment}, where we list the average gaps between our results of Tree-APP/LKH (resp., our results of Path-APP/LKH) and the {pbks}, and the average running times under different settings of $\alpha$. Note that the previous algorithm in~\cite{HarksKM13} achieves an average gap of $18.8\%$, which reduces to $10.0\%$ when further utilizing LKH, as shown in Table~\ref{insinfo}.

Under $\alpha=0.1$, our algorithms yield larger average gaps. The reason may be that under $\alpha=0.1$ the greedy algorithm of UFL opens almost all depots for each instance, making it too expensive. For other values of $\alpha$, we can obtain better results. Notably, as we increase the value of $\alpha$, the gaps exhibit a significant decrease, potentially because the greedy algorithm opens fewer depots. Especially, when setting $\alpha=0.4$ and $\alpha=0.7$, the gaps become $11.64\%$ and $16.76\%$ for Tree-APP and Path-APP, respectively, both outperforming the previous $18.8\%$. Moreover, when setting $\alpha=0.4$ and $\alpha=0.7$, the gaps become $9.44\%$ and $10.77\%$ for Tree-LKH and Path-LKH, respectively, with Tree-LKH surpassing the previous $10.0\%$. However, when further increasing $\alpha$, the gaps begin to increase, likely because the number of opened depots becomes too small.

For the running time, we can see that our algorithms run very fast. 
Under different values of $\alpha$, the running time of each instance on average is at most 0.083s for both Tree-APP and Path-APP, 0.356s for Tree-LKH, and 1.578s for Path-LKH.
The running time of Path-LKH is slow, mainly because it uses the LKH algorithm to find a Hamiltonian cycle on all customers. 
Note that the running time of the previous algorithm~\cite{HarksKM13}, even when using the LKH algorithm, is about 0.01 seconds, which is slightly faster than ours.

We conclude that when setting $\alpha$ suitably, our algorithms can achieve better results than the previous approximation algorithm, and the quality of solutions is also much closer to optimality than the provable approximation ratio.
Although Path-Alg has a better theoretical approximation ratio, the implementations of Tree-Alg perform slightly better on the tested instances in practice. The reason may be that in the worst-case, Tree-Alg delivers only about $k/2$ of demand on each tour but it delivers close to $k$ of demand on each tour in practical. 

\subsubsection{The details}
We summarize the best found results of our algorithms for each instance under $\alpha\in\{0.1,...,1.5\}$ in Table~\ref{experiment+}, where the column `Result' lists the best results of our algorithms under $\alpha\in\{0.1,...,1.5\}$, the column `Gap' shows the gaps between the best results and the corresponding pbks, the column `$\alpha$' presents the minimum values of $\alpha$ at which the algorithms attained the best results, and the last line calculates the average gaps. We can see that if we focus on the best results of our algorithms under $\alpha\in\{0.1,...,1.5\}$, the average gap of Tree-APP achieves $10.55\%$ and of Path-APP achieves $15.20\%$, and the average gap of Tree-LKH achieves $8.38\%$ and of Path-LKH achieves $10.15\%$. The detailed results of Tree-APP/LKH and Path-APP/LKH under each $\alpha\in\{0.1,...,1.5\}$ can be found in the appendix.

\begin{table}[H]
\centering
\resizebox{0.98\textwidth}{!}{
\begin{tabular}{|c|c|c|c|c|c|c|c|c|c|c|c|c|}
\hline
\multirow{2}*{\textbf{Names}} & \multicolumn{3}{c|}{\textbf{Tree-APP}} & \multicolumn{3}{c|}{\textbf{Tree-LKH}} & \multicolumn{3}{c|}{\textbf{Path-APP}} & \multicolumn{3}{c|}{\textbf{Path-LKH}} \\
\cline{2-13}
& \textbf{Result} & \textbf{Gap} & \boldmath $\alpha$ & \textbf{Result} & \textbf{Gap} & \boldmath $\alpha$ & \textbf{Result} & \textbf{Gap} & \boldmath $\alpha$ & \textbf{Result} & \textbf{Gap} & \boldmath $\alpha$ \\
\hline
Chr69-100x10 & 964.04 & 0.1438 & 0.6 & 919.78 & 0.0912 & 0.6 & 1071.73 & 0.2715 & 0.6 & 937.31 & 0.1120 & 0.6\\
Chr69-50x5 & 653.14 & 0.1548 & 0.5 & 629.28 & 0.1126 & 0.5 & 727.90 & 0.2870 & 1.0 & 671.93 & 0.1880 & 1.0\\
Chr69-75x10 & 928.63 & 0.0778 & 0.8 & 909.71 & 0.0559 & 0.8 & 1024.80 & 0.1895 & 0.2 & 951.69 & 0.1046 & 0.2\\
Gas67-22x5 & 612.80 & 0.0473 & 0.3 & 597.57 & 0.0213 & 0.3 & 695.57 & 0.1888 & 0.3 & 671.84 & 0.1483 & 0.3\\
Gas67-29x5 & 562.12 & 0.0977 & 0.3 & 547.07 & 0.0683 & 0.3 & 624.96 & 0.2204 & 0.5 & 576.60 & 0.1260 & 0.5\\
Gas67-32x5 & 584.95 & 0.0405 & 0.9 & 580.32 & 0.0323 & 0.9 & 680.68 & 0.2108 & 0.9 & 646.15 & 0.1494 & 0.9\\
Gas67-32x5-2 & 541.35 & 0.0735 & 0.6 & 535.26 & 0.0614 & 0.6 & 525.39 & 0.0418 & 0.6 & 519.24 & 0.0297 & 0.6\\
Gas67-36x5 & 517.98 & 0.1251 & 0.3 & 502.19 & 0.0908 & 0.3 & 558.24 & 0.2125 & 0.5 & 532.43 & 0.1565 & 0.3\\
Min92-27x5 & 3298.37 & 0.0772 & 0.3 & 3257.83 & 0.0640 & 0.3 & 3328.57 & 0.0871 & 0.3 & 3273.46 & 0.0691 & 0.3\\
P111112 & 1627.45 & 0.1084 & 0.7 & 1586.88 & 0.0807 & 0.7 & 1689.48 & 0.1506 & 0.7 & 1592.75 & 0.0847 & 0.7\\
P111122 & 1620.17 & 0.1180 & 0.9 & 1590.82 & 0.0978 & 0.9 & 1706.30 & 0.1775 & 0.6 & 1666.16 & 0.1498 & 0.6\\
P111212 & 1499.94 & 0.0741 & 0.5 & 1457.91 & 0.0440 & 0.5 & 1682.31 & 0.2047 & 0.6 & 1556.77 & 0.1148 & 0.5\\
P111222 & 1657.59 & 0.1573 & 0.4 & 1601.88 & 0.1184 & 0.4 & 1766.30 & 0.2332 & 0.4 & 1648.72 & 0.1512 & 0.4\\
P112112 & 1358.22 & 0.1634 & 0.4 & 1319.64 & 0.1303 & 0.4 & 1347.73 & 0.1544 & 0.4 & 1308.09 & 0.1204 & 0.4\\
P112122 & 1238.88 & 0.1235 & 0.4 & 1223.69 & 0.1098 & 0.4 & 1294.96 & 0.1744 & 0.4 & 1263.35 & 0.1457 & 0.4\\
P112212 & 853.06 & 0.0745 & 0.1 & 845.21 & 0.0646 & 0.1 & 869.97 & 0.0958 & 0.6 & 836.18 & 0.0532 & 0.6\\
P112222 & 792.83 & 0.0886 & 0.1 & 785.98 & 0.0792 & 0.1 & 803.98 & 0.1039 & 0.3 & 795.04 & 0.0917 & 0.3\\
P113112 & 1420.15 & 0.1467 & 0.2 & 1380.00 & 0.1143 & 0.2 & 1342.14 & 0.0837 & 0.2 & 1291.34 & 0.0427 & 0.2\\
P113122 & 1353.86 & 0.0863 & 0.6 & 1327.25 & 0.0650 & 0.6 & 1461.53 & 0.1727 & 0.6 & 1407.58 & 0.1294 & 0.6\\
P113212 & 927.76 & 0.0282 & 0.1 & 921.73 & 0.0215 & 0.1 & 979.66 & 0.0857 & 0.1 & 945.92 & 0.0483 & 0.1\\
P113222 & 1156.01 & 0.1319 & 0.3 & 1145.93 & 0.1221 & 0.3 & 1163.35 & 0.1391 & 0.3 & 1134.93 & 0.1113 & 0.3\\
P121112 & 2488.88 & 0.0908 & 0.5 & 2443.27 & 0.0708 & 0.5 & 2653.99 & 0.1632 & 0.6 & 2501.80 & 0.0965 & 0.6\\
P121122 & 2429.40 & 0.1116 & 0.6 & 2390.10 & 0.0936 & 0.6 & 2581.16 & 0.1811 & 0.4 & 2442.89 & 0.1178 & 0.4\\
P121212 & 2459.18 & 0.1004 & 0.5 & 2381.71 & 0.0658 & 0.5 & 2609.33 & 0.1676 & 0.5 & 2434.48 & 0.0894 & 0.5\\
P121222 & 2401.74 & 0.0630 & 0.6 & 2348.73 & 0.0395 & 0.6 & 2641.45 & 0.1691 & 0.5 & 2487.35 & 0.1009 & 0.5\\
P122112 & 2408.92 & 0.1461 & 0.4 & 2362.30 & 0.1239 & 0.4 & 2350.18 & 0.1182 & 0.4 & 2243.22 & 0.0673 & 0.5\\
P122122 & 1914.30 & 0.1198 & 0.7 & 1866.88 & 0.0921 & 0.7 & 2013.06 & 0.1776 & 0.4 & 1936.71 & 0.1329 & 0.4\\
P122212 & 1587.52 & 0.0818 & 1.1 & 1566.68 & 0.0676 & 1.1 & 1597.22 & 0.0884 & 1.1 & 1529.33 & 0.0421 & 1.1\\
P122222 & 1181.65 & 0.0893 & 0.1 & 1168.06 & 0.0768 & 0.1 & 1179.49 & 0.0873 & 0.1 & 1140.18 & 0.0511 & 0.1\\
P123112 & 2143.79 & 0.0865 & 0.3 & 2099.42 & 0.0640 & 0.3 & 2108.73 & 0.0687 & 0.3 & 2032.33 & 0.0300 & 0.3\\
P123122 & 2131.50 & 0.0891 & 0.4 & 2089.97 & 0.0679 & 0.4 & 2240.06 & 0.1445 & 0.4 & 2133.93 & 0.0903 & 0.4\\
P123212 & 1999.03 & 0.1288 & 0.1 & 1982.45 & 0.1194 & 0.1 & 1964.92 & 0.1095 & 0.1 & 1916.65 & 0.0823 & 0.8\\
P123222 & 1568.78 & 0.1257 & 0.1 & 1548.63 & 0.1113 & 0.1 & 1535.53 & 0.1019 & 0.1 & 1497.27 & 0.0744 & 0.1\\
P131112 & 2172.39 & 0.1638 & 0.4 & 2108.91 & 0.1298 & 0.4 & 2339.39 & 0.2532 & 0.9 & 2206.60 & 0.1821 & 0.4\\
P131122 & 1996.73 & 0.0841 & 0.6 & 1966.13 & 0.0675 & 0.6 & 2207.36 & 0.1985 & 0.8 & 2115.16 & 0.1484 & 0.8\\
P131212 & 2224.63 & 0.1228 & 0.6 & 2188.29 & 0.1045 & 0.6 & 2290.29 & 0.1559 & 0.6 & 2164.80 & 0.0926 & 0.6\\
P131222 & 1983.94 & 0.0966 & 0.4 & 1917.62 & 0.0599 & 0.4 & 2169.20 & 0.1990 & 0.4 & 2074.71 & 0.1468 & 1.0\\
P132112 & 1641.41 & 0.1334 & 0.7 & 1604.34 & 0.1078 & 0.7 & 1677.30 & 0.1582 & 0.6 & 1620.25 & 0.1188 & 0.3\\
P132122 & 1690.47 & 0.1705 & 0.4 & 1667.31 & 0.1545 & 0.4 & 1608.74 & 0.1139 & 0.4 & 1569.25 & 0.0866 & 0.4\\
P132212 & 1337.87 & 0.1087 & 0.1 & 1329.39 & 0.1017 & 0.1 & 1311.92 & 0.0872 & 0.1 & 1278.44 & 0.0595 & 0.1\\
P132222 & 1008.25 & 0.0819 & 0.1 & 991.36 & 0.0638 & 0.1 & 1025.85 & 0.1008 & 0.1 & 997.53 & 0.0704 & 0.1\\
P133112 & 1919.08 & 0.1290 & 0.6 & 1883.91 & 0.1083 & 0.6 & 1947.17 & 0.1455 & 0.2 & 1873.57 & 0.1022 & 0.2\\
P133122 & 1526.59 & 0.0891 & 0.2 & 1488.73 & 0.0620 & 0.2 & 1599.15 & 0.1408 & 0.7 & 1548.11 & 0.1044 & 0.7\\
P133212 & 1364.01 & 0.1372 & 0.1 & 1346.37 & 0.1225 & 0.1 & 1335.16 & 0.1131 & 0.1 & 1299.64 & 0.0835 & 0.1\\
P133222 & 1219.66 & 0.0580 & 1.5 & 1209.53 & 0.0492 & 1.5 & 1277.35 & 0.1080 & 1.5 & 1233.54 & 0.0700 & 1.5\\
\hline
\textbf{Average Gap} && 0.1055 &&& 0.0838 &&& 0.1520 &&& 0.1015&\\
\hline
\end{tabular}}
\caption{A summary of the best found results of Tree-APP/LKH and Path-APP/LKH, the gaps to pbks, and the value of $\alpha$ for each instance at which the algorithm attained the best result, where the last line calculates the average gaps.
}\label{experiment+}
\end{table}

\section{Conclusion}\label{sec7}
In this paper, we propose two improved approximation algorithms for CFL: one is a 4.169-approximation algorithm for both the unsplittable case and the splittable case, and another is a 4.091-approximation algorithm for the splittable case. Our techniques have the potential to solve more variants of CLR, say the prize-collecting and grouping CLR  mentioned in~\cite{HarksKM13}. 
%Our algorithms may be extended to some variants of CLR, such as prize-collecting and grouping CLR in~\cite{HarksKM13}. 
 
It is known that an $\alpha$-approximation algorithm for CLR also implies an $\alpha$-approximation ratio for MCVRP. Recently, the 4-approximation barrier of MCVRP was broken~\cite{DBLP:conf/cocoon/ZhaoX23}. For further study, it should be interesting to design a better-than-4-approximation algorithm for CLR.

%\THEEndNotes
%\begingroup \parindent 0pt \parskip 0.0ex \def\enotesize{\normalsize} \theendnotes \endgroup

% Appendix here
% Options are (1) APPENDIX (with or without general title) or
%             (2) APPENDICES (if it has more than one unrelated sections)
% Outcomment the appropriate case if necessary
%
% \begin{APPENDIX}{<Title of the Appendix>}
% \end{APPENDIX}
%
%   or
%
% \begin{APPENDICES}
% \section{<Title of Section A>}
% \section{<Title of Section B>}
% etc
% \end{APPENDICES}

\section*{Acknowledgments}
The work is supported by the National Natural Science Foundation of China, under grants 62372095. A preliminary version of this paper~\cite{ijcai24} was presented at the 33rd International Joint Conference on Artificial Intelligence (IJCAI 2024). 
%In the journal version, we include all the omitted proofs. Additionally, we also further slightly improve the approximation ratio for splittable CLR from 4.092 to 4.091, and provide more experimental results using the LKH algorithm.

\section*{Declaration of competing interest} 
The authors declare that they have no known competing financial interests or personal relationships that could have appeared to influence the work reported in this paper.

\bibliographystyle{elsarticle-num} 
\bibliography{main}
\appendix
\clearpage
\section{The Detailed Experimental Results}
The detailed results of Tree-APP and Path-APP can be found in Tables~\ref{res-1}-\ref{res-3}, where we list the gaps between the results and pbks, and the running times for Tree-APP and Path-APP on benchmark instances under $\alpha\in\{0.1,0.2,...,1.5\}$. 
We can see that the average gap of Tree-APP achieves the minimum $0.1164$ when $\alpha=0.4$ and the average gap of Path-APP achieves the minimum $0.1676$ when $\alpha=0.7$.

The detailed results of Tree-LKH and Path-LKH can be found in Tables~\ref{res-1+}-\ref{res-3+}, where we list the gaps between the results and pbks, and the running times for Tree-LKH and Path-LKH on benchmark instances under $\alpha\in\{0.1,0.2,...,1.5\}$. 
We can see that the average gap of Tree-LKH achieves the minimum $0.0944$ when $\alpha=0.4$ and the average gap of Path-LKH achieves the minimum $0.1077$ when $\alpha=0.7$.

\begin{table}[ht]
\centering
\resizebox{1.2\textwidth}{!}{
\begin{tabular}{|c|c|c|c|c|c|c|c|c|c|c|c|c|c|c|c|c|c|c|c|c|}
\hline
\multirow{3}*{\textbf{Names}} & \multicolumn{4}{c|}{\boldmath $\alpha=0.1$} & \multicolumn{4}{c|}{\boldmath $\alpha=0.2$} & \multicolumn{4}{c|}{\boldmath $\alpha=0.3$} & \multicolumn{4}{c|}{\boldmath $\alpha=0.4$} & \multicolumn{4}{c|}{\boldmath $\alpha=0.5$} \\
\cline{2-21}
%\multirow{1.5}*{\textbf{Names}} 
& \multicolumn{2}{c|}{\textbf{Tree-APP}} & \multicolumn{2}{c|}{\textbf{Path-APP}} & \multicolumn{2}{c|}{\textbf{Tree-APP}} & \multicolumn{2}{c|}{\textbf{Path-APP}} & \multicolumn{2}{c|}{\textbf{Tree-APP}} & \multicolumn{2}{c|}{\textbf{Path-APP}} & \multicolumn{2}{c|}{\textbf{Tree-APP}} & \multicolumn{2}{c|}{\textbf{Path-APP}} & \multicolumn{2}{c|}{\textbf{Tree-APP}} & \multicolumn{2}{c|}{\textbf{Path-APP}}\\
\cline{2-21}
& \textbf{Gap} & \textbf{Time} & \textbf{Gap} & \textbf{Time} & \textbf{Gap} & \textbf{Time} & \textbf{Gap} & \textbf{Time} & \textbf{Gap} & \textbf{Time} & \textbf{Gap} & \textbf{Time} & \textbf{Gap} & \textbf{Time} & \textbf{Gap} & \textbf{Time} & \textbf{Gap} & \textbf{Time} & \textbf{Gap} & \textbf{Time}\\
\hline
Chr69-100x10 & 0.3478 & 0.060 & 0.4285 & 0.062 & 0.2508 & 0.057 & 0.3445 & 0.054 & 0.2302 & 0.050 & 0.3331 & 0.053 & 0.2026 & 0.047 & 0.3353 & 0.050 & 0.1673 & 0.044 & 0.2863 & 0.050\\
Chr69-50x5 & 0.3173 & 0.043 & 0.4471 & 0.037 & 0.2708 & 0.039 & 0.3961 & 0.038 & 0.1963 & 0.036 & 0.3909 & 0.037 & 0.1963 & 0.036 & 0.3909 & 0.038 & 0.1548 & 0.034 & 0.3216 & 0.037\\
Chr69-75x10 & 0.1613 & 0.055 & 0.2761 & 0.055 & 0.1035 & 0.054 & 0.1894 & 0.050 & 0.1035 & 0.051 & 0.1894 & 0.050 & 0.1035 & 0.055 & 0.1894 & 0.050 & 0.1035 & 0.054 & 0.1894 & 0.050\\
Gas67-22x5 & 0.2296 & 0.031 & 0.4496 & 0.035 & 0.2296 & 0.029 & 0.4496 & 0.034 & 0.0473 & 0.023 & 0.1888 & 0.030 & 0.0473 & 0.024 & 0.1888 & 0.030 & 0.0473 & 0.023 & 0.1888 & 0.030\\
\hline
\end{tabular}}
\caption{Gaps to pbks and running times for Tree-APP and Path-APP on benchmark instances under $\alpha\in\{0.1,0.2,0.3,0.4,0.5\}$, with the last line representing average gaps and average running times.}
\label{res-1}
\end{table}

\begin{table}[ht]
\centering
\resizebox{1.2\textwidth}{!}{
\begin{tabular}{|c|c|c|c|c|c|c|c|c|c|c|c|c|c|c|c|c|c|c|c|c|}
\hline
\multirow{3}*{\textbf{Names}} & \multicolumn{4}{c|}{\boldmath $\alpha=0.6$} & \multicolumn{4}{c|}{\boldmath $\alpha=0.7$} & \multicolumn{4}{c|}{\boldmath $\alpha=0.8$} & \multicolumn{4}{c|}{\boldmath $\alpha=0.9$} & \multicolumn{4}{c|}{\boldmath $\alpha=1.0$} \\
\cline{2-21}
%\multirow{1.5}*{\textbf{Names}} 
& \multicolumn{2}{c|}{\textbf{Tree-APP}} & \multicolumn{2}{c|}{\textbf{Path-APP}} & \multicolumn{2}{c|}{\textbf{Tree-APP}} & \multicolumn{2}{c|}{\textbf{Path-APP}} & \multicolumn{2}{c|}{\textbf{Tree-APP}} & \multicolumn{2}{c|}{\textbf{Path-APP}} & \multicolumn{2}{c|}{\textbf{Tree-APP}} & \multicolumn{2}{c|}{\textbf{Path-APP}} & \multicolumn{2}{c|}{\textbf{Tree-APP}} & \multicolumn{2}{c|}{\textbf{Path-APP}}\\
\cline{2-21}
& \textbf{Gap} & \textbf{Time} & \textbf{Gap} & \textbf{Time} & \textbf{Gap} & \textbf{Time} & \textbf{Gap} & \textbf{Time} & \textbf{Gap} & \textbf{Time} & \textbf{Gap} & \textbf{Time} & \textbf{Gap} & \textbf{Time} & \textbf{Gap} & \textbf{Time} & \textbf{Gap} & \textbf{Time} & \textbf{Gap} & \textbf{Time}\\
\hline
Chr69-100x10 & 0.1438 & 0.047 & 0.2715 & 0.051 & 0.1882 & 0.046 & 0.3019 & 0.058 & 0.1882 & 0.046 & 0.3019 & 0.050 & 0.2192 & 0.043 & 0.3148 & 0.054 & 0.1674 & 0.042 & 0.3165 & 0.049\\
Chr69-50x5 & 0.1548 & 0.033 & 0.3216 & 0.035 & 0.1548 & 0.034 & 0.3216 & 0.034 & 0.1548 & 0.033 & 0.3216 & 0.033 & 0.1548 & 0.034 & 0.3216 & 0.035 & 0.2127 & 0.032 & 0.2870 & 0.033\\
Chr69-75x10 & 0.1035 & 0.053 & 0.1894 & 0.050 & 0.1009 & 0.050 & 0.2112 & 0.056 & 0.0778 & 0.048 & 0.2574 & 0.055 & 0.0778 & 0.046 & 0.2574 & 0.052 & 0.0778 & 0.046 & 0.2574 & 0.050\\
Gas67-22x5 & 0.0473 & 0.024 & 0.1888 & 0.031 & 0.0473 & 0.023 & 0.1888 & 0.031 & 0.0473 & 0.024 & 0.1888 & 0.030 & 0.0473 & 0.023 & 0.1888 & 0.029 & 0.0473 & 0.023 & 0.1888 & 0.029\\
\hline
\end{tabular}}
\caption{Gaps to pbks and running times for Tree-APP and Path-APP on benchmark instances under $\alpha\in\{0.6,0.7,0.8,0.9,1.0\}$, with the last line representing average gaps and average running times.}
\label{res-2}
\end{table}

\begin{table}[ht]
\centering
\resizebox{1.2\textwidth}{!}{
\begin{tabular}{|c|c|c|c|c|c|c|c|c|c|c|c|c|c|c|c|c|c|c|c|c|}
\hline
\multirow{3}*{\textbf{Names}} & \multicolumn{4}{c|}{\boldmath $\alpha=1.1$} & \multicolumn{4}{c|}{\boldmath $\alpha=1.2$} & \multicolumn{4}{c|}{\boldmath $\alpha=1.3$} & \multicolumn{4}{c|}{\boldmath $\alpha=1.4$} & \multicolumn{4}{c|}{\boldmath $\alpha=1.5$} \\
\cline{2-21}
%\multirow{1.5}*{\textbf{Names}} 
& \multicolumn{2}{c|}{\textbf{Tree-APP}} & \multicolumn{2}{c|}{\textbf{Path-APP}} & \multicolumn{2}{c|}{\textbf{Tree-APP}} & \multicolumn{2}{c|}{\textbf{Path-APP}} & \multicolumn{2}{c|}{\textbf{Tree-APP}} & \multicolumn{2}{c|}{\textbf{Path-APP}} & \multicolumn{2}{c|}{\textbf{Tree-APP}} & \multicolumn{2}{c|}{\textbf{Path-APP}} & \multicolumn{2}{c|}{\textbf{Tree-APP}} & \multicolumn{2}{c|}{\textbf{Path-APP}}\\
\cline{2-21}
& \textbf{Gap} & \textbf{Time} & \textbf{Gap} & \textbf{Time} & \textbf{Gap} & \textbf{Time} & \textbf{Gap} & \textbf{Time} & \textbf{Gap} & \textbf{Time} & \textbf{Gap} & \textbf{Time} & \textbf{Gap} & \textbf{Time} & \textbf{Gap} & \textbf{Time} & \textbf{Gap} & \textbf{Time} & \textbf{Gap} & \textbf{Time}\\
\hline
Chr69-100x10 & 0.1674 & 0.045 & 0.3165 & 0.056 & 0.1674 & 0.046 & 0.3165 & 0.054 & 0.1674 & 0.047 & 0.3165 & 0.055 & 0.1674 & 0.044 & 0.3165 & 0.052 & 0.1674 & 0.044 & 0.3165 & 0.052\\
Chr69-50x5 & 0.2127 & 0.034 & 0.2870 & 0.034 & 0.2127 & 0.034 & 0.2870 & 0.034 & 0.2127 & 0.034 & 0.2870 & 0.040 & 0.2127 & 0.034 & 0.2870 & 0.033 & 0.2127 & 0.033 & 0.2870 & 0.036\\
Chr69-75x10 & 0.0778 & 0.050 & 0.2574 & 0.054 & 0.0778 & 0.049 & 0.2574 & 0.055 & 0.0778 & 0.047 & 0.2574 & 0.052 & 0.0778 & 0.048 & 0.2574 & 0.054 & 0.0778 & 0.048 & 0.2574 & 0.052\\
Gas67-22x5 & 0.0473 & 0.025 & 0.1888 & 0.029 & 0.0473 & 0.024 & 0.1888 & 0.033 & 0.0473 & 0.023 & 0.1888 & 0.032 & 0.0473 & 0.023 & 0.1888 & 0.030 & 0.0473 & 0.023 & 0.1888 & 0.031\\
\hline
\end{tabular}}
\caption{Gaps to pbks and running times for Tree-APP and Path-APP on benchmark instances under $\alpha\in\{1.1,1.2,1.3,1.4,1.5\}$, with the last line representing average gaps and average running times.}
\label{res-3}
\end{table}

\begin{table}[ht]
\centering
\resizebox{1.2\textwidth}{!}{
\begin{tabular}{|c|c|c|c|c|c|c|c|c|c|c|c|c|c|c|c|c|c|c|c|c|}
\hline
\multirow{3}*{\textbf{Names}} & \multicolumn{4}{c|}{\boldmath $\alpha=0.1$} & \multicolumn{4}{c|}{\boldmath $\alpha=0.2$} & \multicolumn{4}{c|}{\boldmath $\alpha=0.3$} & \multicolumn{4}{c|}{\boldmath $\alpha=0.4$} & \multicolumn{4}{c|}{\boldmath $\alpha=0.5$} \\
\cline{2-21}
%\multirow{1.5}*{\textbf{Names}} 
& \multicolumn{2}{c|}{\textbf{Tree-LKH}} & \multicolumn{2}{c|}{\textbf{Path-LKH}} & \multicolumn{2}{c|}{\textbf{Tree-LKH}} & \multicolumn{2}{c|}{\textbf{Path-LKH}} & \multicolumn{2}{c|}{\textbf{Tree-LKH}} & \multicolumn{2}{c|}{\textbf{Path-LKH}} & \multicolumn{2}{c|}{\textbf{Tree-LKH}} & \multicolumn{2}{c|}{\textbf{Path-LKH}} & \multicolumn{2}{c|}{\textbf{Tree-LKH}} & \multicolumn{2}{c|}{\textbf{Path-LKH}}\\
\cline{2-21}
& \textbf{Gap} & \textbf{Time} & \textbf{Gap} & \textbf{Time} & \textbf{Gap} & \textbf{Time} & \textbf{Gap} & \textbf{Time} & \textbf{Gap} & \textbf{Time} & \textbf{Gap} & \textbf{Time} & \textbf{Gap} & \textbf{Time} & \textbf{Gap} & \textbf{Time} & \textbf{Gap} & \textbf{Time} & \textbf{Gap} & \textbf{Time}\\
\hline
Chr69-100x10 & 0.3243 & 0.244 & 0.3158 & 0.516 & 0.2222 & 0.200 & 0.2034 & 0.502 & 0.1783 & 0.217 & 0.2014 & 0.374 & 0.1499 & 0.190 & 0.1607 & 0.432 & 0.0990 & 0.194 & 0.1270 & 0.405\\
Chr69-50x5 & 0.2737 & 0.160 & 0.2937 & 0.323 & 0.2297 & 0.150 & 0.2520 & 0.283 & 0.1525 & 0.159 & 0.2259 & 0.190 & 0.1525 & 0.138 & 0.2259 & 0.189 & 0.1126 & 0.136 & 0.1503 & 0.173\\
Chr69-75x10 & 0.1366 & 0.228 & 0.1311 & 0.402 & 0.0712 & 0.200 & 0.1120 & 0.237 & 0.0712 & 0.239 & 0.1120 & 0.260 & 0.0712 & 0.215 & 0.1120 & 0.257 & 0.0712 & 0.216 & 0.1120 & 0.254\\
Gas67-22x5 & 0.2170 & 0.069 & 0.2173 & 0.088 & 0.2170 & 0.086 & 0.2173 & 0.077 & 0.0213 & 0.061 & 0.1851 & 0.096 & 0.0213 & 0.049 & 0.1851 & 0.086 & 0.0213 & 0.058 & 0.1851 & 0.093\\
\hline
\end{tabular}}
\caption{Gaps to pbks and running times for Tree-LKH and Path-LKH on benchmark instances under $\alpha\in\{0.1,0.2,0.3,0.4,0.5\}$, with the last line representing average gaps and average running times.}
\label{res-1+}
\end{table}

\begin{table}[ht]
\centering
\resizebox{1.2\textwidth}{!}{
\begin{tabular}{|c|c|c|c|c|c|c|c|c|c|c|c|c|c|c|c|c|c|c|c|c|}
\hline
\multirow{3}*{\textbf{Names}} & \multicolumn{4}{c|}{\boldmath $\alpha=0.6$} & \multicolumn{4}{c|}{\boldmath $\alpha=0.7$} & \multicolumn{4}{c|}{\boldmath $\alpha=0.8$} & \multicolumn{4}{c|}{\boldmath $\alpha=0.9$} & \multicolumn{4}{c|}{\boldmath $\alpha=1.0$} \\
\cline{2-21}
%\multirow{1.5}*{\textbf{Names}} 
& \multicolumn{2}{c|}{\textbf{Tree-LKH}} & \multicolumn{2}{c|}{\textbf{Path-LKH}} & \multicolumn{2}{c|}{\textbf{Tree-LKH}} & \multicolumn{2}{c|}{\textbf{Path-LKH}} & \multicolumn{2}{c|}{\textbf{Tree-LKH}} & \multicolumn{2}{c|}{\textbf{Path-LKH}} & \multicolumn{2}{c|}{\textbf{Tree-LKH}} & \multicolumn{2}{c|}{\textbf{Path-LKH}} & \multicolumn{2}{c|}{\textbf{Tree-LKH}} & \multicolumn{2}{c|}{\textbf{Path-LKH}}\\
\cline{2-21}
& \textbf{Gap} & \textbf{Time} & \textbf{Gap} & \textbf{Time} & \textbf{Gap} & \textbf{Time} & \textbf{Gap} & \textbf{Time} & \textbf{Gap} & \textbf{Time} & \textbf{Gap} & \textbf{Time} & \textbf{Gap} & \textbf{Time} & \textbf{Gap} & \textbf{Time} & \textbf{Gap} & \textbf{Time} & \textbf{Gap} & \textbf{Time}\\
\hline
Chr69-100x10 & 0.0912 & 0.196 & 0.1408 & 0.379 & 0.1166 & 0.219 & 0.1434 & 0.454 & 0.1166 & 0.188 & 0.1434 & 0.461 & 0.1713 & 0.209 & 0.1713 & 0.401 & 0.1084 & 0.186 & 0.1393 & 0.520\\
Chr69-50x5 & 0.1126 & 0.126 & 0.1503 & 0.169 & 0.1126 & 0.128 & 0.1503 & 0.173 & 0.1126 & 0.124 & 0.1503 & 0.182 & 0.1126 & 0.143 & 0.1503 & 0.167 & 0.1710 & 0.136 & 0.1677 & 0.172\\
Chr69-75x10 & 0.0712 & 0.220 & 0.1120 & 0.249 & 0.0749 & 0.241 & 0.1431 & 0.289 & 0.0558 & 0.198 & 0.0887 & 0.246 & 0.0558 & 0.231 & 0.0887 & 0.239 & 0.0558 & 0.200 & 0.0887 & 0.241\\
Gas67-22x5 & 0.0213 & 0.055 & 0.1851 & 0.083 & 0.0213 & 0.058 & 0.1851 & 0.085 & 0.0213 & 0.051 & 0.1851 & 0.085 & 0.0213 & 0.065 & 0.1851 & 0.087 & 0.0213 & 0.055 & 0.1851 & 0.089\\
\hline
\end{tabular}}
\caption{Gaps to pbks and running times for Tree-LKH and Path-LKH on benchmark instances under $\alpha\in\{0.6,0.7,0.8,0.9,1.0\}$, with the last line representing average gaps and average running times.}
\label{res-2+}
\end{table}

\begin{table}[ht]
\centering
\resizebox{1.2\textwidth}{!}{
\begin{tabular}{|c|c|c|c|c|c|c|c|c|c|c|c|c|c|c|c|c|c|c|c|c|}
\hline
\multirow{3}*{\textbf{Names}} & \multicolumn{4}{c|}{\boldmath $\alpha=1.1$} & \multicolumn{4}{c|}{\boldmath $\alpha=1.2$} & \multicolumn{4}{c|}{\boldmath $\alpha=1.3$} & \multicolumn{4}{c|}{\boldmath $\alpha=1.4$} & \multicolumn{4}{c|}{\boldmath $\alpha=1.5$} \\
\cline{2-21}
%\multirow{1.5}*{\textbf{Names}} 
& \multicolumn{2}{c|}{\textbf{Tree-LKH}} & \multicolumn{2}{c|}{\textbf{Path-LKH}} & \multicolumn{2}{c|}{\textbf{Tree-LKH}} & \multicolumn{2}{c|}{\textbf{Path-LKH}} & \multicolumn{2}{c|}{\textbf{Tree-LKH}} & \multicolumn{2}{c|}{\textbf{Path-LKH}} & \multicolumn{2}{c|}{\textbf{Tree-LKH}} & \multicolumn{2}{c|}{\textbf{Path-LKH}} & \multicolumn{2}{c|}{\textbf{Tree-LKH}} & \multicolumn{2}{c|}{\textbf{Path-LKH}}\\
\cline{2-21}
& \textbf{Gap} & \textbf{Time} & \textbf{Gap} & \textbf{Time} & \textbf{Gap} & \textbf{Time} & \textbf{Gap} & \textbf{Time} & \textbf{Gap} & \textbf{Time} & \textbf{Gap} & \textbf{Time} & \textbf{Gap} & \textbf{Time} & \textbf{Gap} & \textbf{Time} & \textbf{Gap} & \textbf{Time} & \textbf{Gap} & \textbf{Time}\\
\hline
Chr69-100x10 & 0.1084 & 0.219 & 0.1393 & 0.538 & 0.1084 & 0.181 & 0.1393 & 0.521 & 0.1084 & 0.200 & 0.1393 & 0.530 & 0.1084 & 0.179 & 0.1393 & 0.517 & 0.1084 & 0.192 & 0.1393 & 0.516\\
Chr69-50x5 & 0.1710 & 0.155 & 0.1677 & 0.176 & 0.1710 & 0.131 & 0.1677 & 0.176 & 0.1710 & 0.147 & 0.1677 & 0.169 & 0.1710 & 0.128 & 0.1677 & 0.168 & 0.1710 & 0.141 & 0.1677 & 0.168\\
Chr69-75x10 & 0.0558 & 0.232 & 0.0887 & 0.256 & 0.0558 & 0.201 & 0.0887 & 0.251 & 0.0558 & 0.231 & 0.0887 & 0.252 & 0.0558 & 0.204 & 0.0887 & 0.250 & 0.0558 & 0.216 & 0.0887 & 0.248\\
Gas67-22x5 & 0.0213 & 0.066 & 0.1851 & 0.088 & 0.0213 & 0.059 & 0.1851 & 0.093 & 0.0213 & 0.061 & 0.1851 & 0.086 & 0.0213 & 0.056 & 0.1851 & 0.086 & 0.0213 & 0.058 & 0.1851 & 0.083\\
\hline
\end{tabular}}
\caption{Gaps to pbks and running times for Tree-LKH and Path-LKH on benchmark instances under $\alpha\in\{1.1,1.2,1.3,1.4,1.5\}$, with the last line representing average gaps and average running times.}
\label{res-3+}
\end{table}
\end{document}